\documentclass{amsart}

\usepackage{amssymb, graphicx, amsrefs, mathrsfs}
\usepackage{color}%cyan,red;magenta,green;yellow,blue
\usepackage{mathtools}
\usepackage{pgfplots}
\usepackage{tikz}
\usepackage{tikz-network}
\usetikzlibrary{positioning}

\usetikzlibrary{decorations.markings, arrows.meta}
\tikzset{
mid arrow/.style={postaction={decorate, decoration={
markings,
mark=at position .6 with {\arrow{Straight Barb}}
}}},
}

\usepackage{eqparbox}
\usepackage{hyperref}

\copyrightinfo{2025}{Luis Hernandez, Sean Ku, Jun Masamune,
Genevieve Romanelli and Rados{\l}aw K. Wojciechowski}

\newtheorem{theorem}{Theorem}[section]
\newtheorem{lemma}[theorem]{Lemma}

\newtheorem{corollary}[theorem]{Corollary}

\theoremstyle{definition}
\newtheorem{definition}[theorem]{Definition}
\newtheorem{example}[theorem]{Example}

\theoremstyle{remark}
\newtheorem{remark}[theorem]{Remark}

\numberwithin{equation}{section}

\newcommand{\as}[1]{\left\langle #1\right\rangle}

\newcommand{\ov}[1]{\overline{ #1}}
\newcommand{\ow}[1]{\widetilde{ #1}}

\newcommand{\N}{\mathbb{N}}
\newcommand{\Z}{\mathbb{Z}}
\newcommand{\R}{\mathbb{R}}

\newcommand{\cp}{\textup{cap}}
\newcommand{\vp}{\varphi}

\newcommand{\si}{\sigma}

\newcommand{\Deg}{\textup{Deg}}

\newcommand{\pt}{\partial}
\newcommand{\De}{\Delta}
\newcommand{\al}{\alpha}
\newcommand{\be}{\beta}
\newcommand{\ka}{\kappa}
\newcommand{\Ga}{\Gamma}
\newcommand{\Q}{\mathcal{Q}}
\newcommand{\D}{\mathcal{D}}
\newcommand{\F}{\mathcal{F}}

\newcommand{\A}{\mathscr{A}}

\newcommand{\QD}{Q^{(D)}}
\newcommand{\QN}{Q^{(N)}}

\newcommand{\PD}{P_t^{(D)}}
\newcommand{\PN}{P_t^{(N)}}

\providecommand{\eat}[1]{}

\newcommand{\Hm}[1]{\leavevmode{\marginpar{\tiny%
$\hbox to 0mm{\hspace*{-0.5mm}$\leftarrow$\hss}%
\vcenter{\vrule depth 0.1mm height 0.1mm width \the\marginparwidth}%
\hbox to 0mm{\hss$\rightarrow$\hspace*{-0.5mm}}$\\\relax\raggedright
#1}}}

\begin{document}

\title[Form uniqueness]{Form uniqueness for graphs 
with weakly spherically symmetric ends}

\author[Hernandez]{Luis Hernandez}
\address{L.~Hernandez, York College of the City University of New York, 94-20 Guy R. Brewer Blvd., Jamaica, NY 11451}
\email{luis.hernandez2@yorkmail.cuny.edu}

\author[Ku]{Sean Ku}
\address{S.~Ku, UCLA Department of Mathematics, 520 Portola Plaza, Los Angeles, CA 90095}
\email{sk8980@nyu.edu, seanku@math.ucla.edu}

\author[Masamune]{Jun Masamune}
\address{J.~Masamune, Mathematical Institute, Tohoku University, 6-3, Aramaki Aza-Aoba, Aoba-ku, Sendai 980-8578, Japan}
\email{jun.masamune.c3@tohoku.ac.jp}

\author[Romanelli]{Genevieve Romanelli}
\address{G.~Romanelli, Graduate Center of the City University of New York, 365 Fifth Avenue, New York, NY, 10016}
\email{gromanelli@gradcenter.cuny.edu}

\author[Wojciechowski]{Rados{\l}aw K. Wojciechowski}
\address{R.~K.~Wojciechowski, Graduate Center of the City University of New York, 365 Fifth Avenue, New York, NY, 10016}
\address{York College of the City University of New York, 94-20 Guy R. Brewer Blvd., Jamaica, NY 11451}
\email{rwojciechowski@gc.cuny.edu}


\date{\today}
\thanks{This work is supported by the Queens Experiences in 
Discrete Mathematics (QED) REU program funded by the 
National Science Foundation, Award Number DMS 2150251.}
\thanks{J.~M.~is additionally supported 
by JSPS KAKENHI Grant Number 23H03798 and LUPICIA CO., LTD}
\thanks{R.~K.~W.~is additionally supported by 
a Travel Support for Mathematicians Gift funded by the Simons Foundation
and a Department Chair Research Account established by the PSC and RF CUNY}

\begin{abstract}
We give characterizations for the failure of form uniqueness on
weakly spherically symmetric graphs. 
The first characterization is in terms of the graph structure, 
the second involves the capacity
of a Cauchy boundary. We also discuss the stability of form uniqueness and
give characterizations for graphs which, following the removal of a set,
consist of a disjoint union of weakly spherically symmetric graphs.
\end{abstract} 

\maketitle
\tableofcontents

\section{Introduction}
Form (equivalently, Markov) uniqueness is one of the fundamental questions
for infinite weighted graphs. This property has to do with boundary conditions
at infinity. For the form uniqueness viewpoint, one considers various
closed restrictions of the energy form to subspaces of the square summable functions and asks under which conditions all form restrictions agree, 
see \cites{HKLW12, HKMW13, KLW21, LSW21, KMW25, Sch20b}. 
An equivalent viewpoint is given
via consideration of Markov realizations of the associated Laplacian, that is,
realizations of the operator whose corresponding form is a Dirichlet form. 
For Markov uniqueness, 
one asks the question of when there is a unique Markov realization of
the Laplacian restricted to the finitely supported functions.
See \cite{HKLW12, Sch20a, Sch20b}
for the equivalence of the form and Markov uniqueness viewpoints for weighted graphs.
For a third viewpoint, we note that Markov uniqueness also means that the space 
has a unique Markov process associated with the 
energy, see \cite{FOT94} for general background.

In this paper we study form uniqueness for weakly spherically symmetric graphs as 
introduced in \cite{KLW13} and generalized in \cite{BG15}. 
In parallel to transience and stochastic incompleteness, we 
give a characterization for the failure of form uniqueness for such graphs in terms
of the edge weights, killing term, and vertex measure (Theorem~\ref{t:form_uniqueness_weights}). It is worth
noting that, for graphs of finite measure with no killing term, the
notions of transience, stochastic incompleteness and failure
of form uniqueness coincide, see \cite{Sch17}. Theorem~\ref{t:form_uniqueness_weights} 
shows that, for weakly spherically symmetric graphs with no killing term,
the only way that form uniqueness fails is when the graph has 
finite measure and is transient (equivalently, stochastically incomplete). 

We also look at form uniqueness 
in the context of a capacity of a boundary at infinity
as studied in \cite{HKMW13} and generalized in \cite{Sch20b}. Here,
we prove that the failure of form uniqueness for weakly spherically symmetric
graphs is equivalent to the existence
of a Cauchy boundary which has positive and finite capacity (Theorem~\ref{t:capacity}). 
For this result, we have to work
with strongly intrinsic path metrics.

We then discuss the stability of the failure of form uniqueness.
In this context we show that if a graph can be decomposed
into two subgraphs which are not too strongly connected, then
the graph is not form unique if and only if one of the subgraphs
is not form unique (Theorem~\ref{t:stability}).
For related results for stochastic incompleteness see \cite{Hua11, KL12},
for the failure of essential self-adjointness see \cite{IKMW25}.
We then apply this to the case of graphs which, following the removal
of a set of vertices, are disjoint unions of weakly spherically symmetric
graphs, i.e., graphs with weakly spherically symmetric ends. 
For these graphs, we prove that the failure of form uniqueness means that
at least one of the ends is not form unique (Theorem~\ref{t:ends}).

A more widely studied property on weighted graphs 
is essential self-adjointness.
This concerns the uniqueness
of self-adjoint extensions of the restriction of the Laplacian to the finitely supported function. As 
the uniqueness of self-adjoint extensions will imply
the uniqueness of Markov realizations, this is a stronger property than form uniqueness.
Essential self-adjointness of the Laplacian on graphs 
has been extensively studied in, for example, \cite{deTT11, JP11, HKMW13, HMW21, IKMW25, KL10, KL12, KMW25, Ku, Mil11,  
Sch20b, Tor10, Web10, Woj08}, among many other works.

In the recent paper \cite{IKMW25}, essential 
self-adjointness is characterized via the graph structure and another notion of capacity for 
birth--death chains which are the easiest class of weakly spherically symmetric graphs.
These characterizations for essential self-adjointness, however, 
do not extend to all weakly spherically symmetric
graphs as we are able to do for form uniqueness in this paper. One technical issue is that, 
for essential self-adjointness, one is not able to restrict to working 
with positive $\al$-harmonic functions but rather has to consider all such functions.
%In terms of stability theory, \cite{IKMW25}
%decomposes the graph and shows that essential self-adjointness can
%be studied by looking at the subgraphs. Here, we give a parallel development
%for form uniqueness.

We also mention that Markov uniqueness
and essential self-adjointness are well-studied properties for manifolds, quantum graphs,
and general Dirichlet forms, see \cite{Che73, FOT94, Gaf51, Gaf54, GM13, HKLMS17, KMN22a, KMN22b, KN21, KN22, KN23, Mas99, Mas05, RS1, RS2, Str83 } 
for a sampling of the literature.

We now discuss the structure of the paper.
In Section~\ref{s:background} we introduce weakly spherically symmetric graphs,
energy forms, Laplacian and averaging operators and state some
general criteria for form uniqueness as well as basic facts
concerning weakly spherically symmetric graphs.
In Section~\ref{s:form_uniqueness} we prove our characterizations of form uniqueness via both the graph structure and capacity. We also 
show how a characterization for the Feller property of the Neumann
semigroup for birth--death chains from \cite{KMW25} extends to all 
weakly spherically symmetric graphs and discuss how the 
criteria for form uniqueness 
are related to characterizations for other properties.

In Section~\ref{s:stability} we discuss stability theory. There we show
how failure of form uniqueness is stable under decomposition of the graph
and then apply this to graphs with symmetric ends. We also give examples
to show how stability can fail, i.e., we give examples of graphs which
are not form unique and then attach a graph so that the entire
graph satisfies form uniqueness.

\section{Background material}\label{s:background}
We start by introducing the setting and giving some general background results.
First, we introduce graphs, then focus on weakly spherically
symmetric graphs and give several classes of examples. We then
introduce forms and graph Laplacians as well as the averaging operator. 
This gives another viewpoint on weak
spherical symmetry. Along the way, we introduce some general
theory for studying form uniqueness via $\al$-harmonic functions.

\subsection{Graphs}
We introduce the general setting following \cite{KL12, KLW21}.
Let $X$ denote a countable set whose elements are called \emph{vertices}. Let $b\colon X \times X \to [0,\infty)$ denote the \emph{edge 
weight} function which satisfies 

\begin{itemize}
\item{\makebox[3cm]{$b(x,x)=0$ \hfill} for all  $x \in X$,}
\item{\makebox[3cm]{$b(x,y)=b(y,x)$ \hfill} for all $x, y \in X,$}
\item{\makebox[3cm]{$\sum_{y \in X}b(x,y)<\infty$ \hfill} for all $x \in X.$}
\end{itemize}
Furthermore, let $c \colon X \to [0,\infty)$ denote the \emph{killing
term} and let $m \colon X \to (0,\infty)$ denote the \emph{vertex measure}.
We call $(b,c)$ a \emph{graph} over $(X,m)$. When $c=0$, we say
that $b$ is a \emph{graph} over $(X,m)$.
In what follows $c$, $m$ and their sum 
will be extended to all subsets of $X$ by additivity. That is, if $A \subseteq X$, then
$$c(A)=\sum_{x \in A}c(x), \qquad m(A) = \sum_{x \in A} m(x), \qquad \textup{and} \qquad (c+m)(A)=\sum_{x \in A}(c+m)(x).$$

Two vertices $x$ and $y$ are \emph{connected} or are \emph{neighbors} or \emph{form an edge} if $b(x,y)>0$. 
We write $x \sim y$ in this case. 
For a vertex $x \in X$, let 
$$\Deg(x) = \frac{1}{m(x)}\left(\sum_{y \in X}b(x,y) + c(x)\right)$$
denote the \emph{weighted vertex degree} which is
the total edge weight of edges incident to $x$ plus the killing
term normalized by the measure.

A graph is \emph{connected} if for any two
vertices $x, y \in X$ there exists a sequence of vertices $(x_k)$ which starts at $x$,
ends at $y$, and so that all subsequent vertices are connected. 
Such a sequence is called
a \emph{path} connecting $x$ and $y$.
A graph is \emph{locally finite} if every vertex has only finitely many neighbors. 

If $(x_k)$ is a path consisting of $n+1$ vertices for $n \in \N$, 
then the \emph{length} of $(x_k)$ is $n$,
i.e., the number of edges in the path. For connected graphs,
we denote the least length of a path connecting
$x$ and $y$ for $x \neq y$ by $d(x,y)$, let $d(x,x)=0$, and call $d$ the \emph{combinatorial graph distance}.

For a subset $O \subseteq X$ and $r \in \N_0$, let
$S_r(O) = \{ x \in X \mid d(x,O)=r\}$ denote the sphere of radius $r$ about $O$ and
$B_r(O) = \{ x \in X \mid d(x,O)\leq r \}$ denote the ball of radius $r$ about $O$.
Generally, we mask the dependence on $O$ and just write $S_r$ for $S_r(O)$
and $B_r$ for $B_r(O)$.
Let
$$\pt B(r) = \sum_{x \in S_r}\sum_{y \in S_{r+1}} b(x,y)$$
denote the \emph{edge boundary of the ball of radius $r$}, which is the sum
of the weights of all edges leaving the ball of radius $r$.

\subsection{Weakly spherically symmetric graphs}
We now focus on the class of graphs that will be of primary interest. 
For this class, we assume that all graphs are connected and locally finite. 
In particular, local finiteness implies that 
all balls and spheres about a finite set are finite. Fix a finite subset
$O \subseteq X$.
The definitions that follow depend on the choice
of $O$, though we will generally suppress this.

Let $C(X)= \{f \colon X \to \R\}$ denote the space of all functions.
A function $f \in C(X)$ is \emph{spherically symmetric with respect
to $O$} if $f(x)=f(y)$ for all $x, y \in S_r=S_r(O)$ and all $r\in \N_0$. We then
write $f(r)$ for the common value of $f$ on $S_r$.

For $x \in S_r$, $r \in \N_0$, let
$$\ka_\pm(x) = \frac{1}{m(x)}\sum_{y \in S_{r \pm 1}}b(x,y)$$
denote the \emph{outer} and \emph{inner vertex degree}, respectively,
where $S_{-1}=\emptyset$.
Furthermore, let $q(x) = c(x)/m(x)$.

\begin{definition}[Weakly spherically symmetric graphs]
A locally finite connected graph $(b,c)$ over $(X,m)$
is \emph{weakly spherically symmetric} if there exists a finite set $O \subseteq X$ such that
$\ka_\pm$ and $q$ are spherically symmetric functions with respect to $O$.
\end{definition}

\begin{remark}[On the definition of weakly spherically symmetric graphs]
Letting $\ka_0(x) = \frac{1}{m(x)}\sum_{y \in S_{r}}b(x,y)$ for $x \in S_r$, 
the weighted vertex degree of $x$ can be decomposed as
$$\Deg(x) = \ka_+(x)+\ka_-(x)+\ka_0(x) + q(x).$$
Thus, the assumption on weak spherical symmetry can be understood as an assumption on
parts of the weighted vertex degree when thinking of the graph as spheres
about $O$.
Furthermore, since we allow for a general measure and edge weight, the notion
of symmetry introduced above is quite weak. In fact, even
in the unweighted setting, this notion is quite weak, see Figure 1 in \cite{KLW13}
for some examples illustrating this.
\end{remark}

%We note that for weakly spherically symmetric
%graphs
%$$\pt B(r)= \ka_+(r)m(S_r) = \ka_-(r+1) m(S_{r+1})$$
%as follows by direct calculations.

\begin{example}[Examples of weakly spherically symmetric graphs] 
We give three standard examples for weakly spherically symmetric graphs
below. 

(1) \textbf{Birth--death chains:}
A graph is called a \emph{birth--death chain}
if $X = \N_0=\{0, 1, 2, \ldots\}$, 
$b(x,y)>0$ if and only if $|x-y|=1$ and $c=0$. Such graphs
are clearly weakly spherically symmetric with respect to $O=\{0\}$ 
as each sphere consists of only one vertex. 
For birth--death chains, note that
$$\ka_+(r)=\frac{b(r,r+1)}{m(r)}, \quad  \ka_-(r)=\frac{b(r,r-1)}{m(r)},
\qquad \textup{and} \qquad \pt B(r)=b(r,r+1).$$

(2) \textbf{Spherically Symmetric Trees:} Let $b(x,y) \in \{0,1\}$,
$m=1$, and $c=0$. Fix a vertex $x_0 \in X$, let $O=\{x_0\}$
and define $b$ so that 
$$\ka_-(r) =1 \qquad \textup{ and } \qquad \ka_+(r)=k(r)$$ 
where $k\colon \N_0 \to \N$ is an arbitrary function which can be thought of as the branching number of vertices on a sphere.
Furthermore, let
$b \vert_{S_r \times S_r}=0$ for $r \in \N$, 
i.e., $\ka_0(x)=0$ for all $x \in X$, thus
ensuring that the resulting graph is a tree and call such graphs
\emph{weakly spherically symmetric trees}.
For such graphs, note that
$$\pt B(r) = \prod_{i=0}^r k(i).$$

(3) \textbf{Anti-trees:} Let $b(x,y) \in \{0,1\}$,
$m=1$, and $c=0$. Fix a vertex $x_0 \in X$ and let $O=\{x_0\}$. Choose an arbitrary
sphere growth function $s \colon \N_0 \to \N$ and define $b$ so that
$$\ka_-(r)=s(r-1) \qquad \textup{ and } \qquad \ka_+(r)=s(r+1).$$
In other words, every vertex in the sphere $S_r$ is connected to 
all vertices in the next sphere $S_{r+1}$ for all $r \in \N_0$. 
We call such graphs \emph{anti-trees}.
The idea for these graphs goes back 
to at least \cite{DK87} and was generalized to the form above in \cite{Woj11}. Anti-trees
have served as useful (counter-)examples
in various places, e.g., \cite{AS23, BKW15, BK13, CLMP20, GHM12, GS13, KR24, 
KS15, KMN22b, KN21, MUW12, Sad16, Web10, Woj11, Woj21}. For anti-trees, note that
$$\pt B(r)=s(r)s(r+1).$$

\end{example}

\subsection{Energy form and Laplacian}
We next introduce the necessary function spaces, forms and associated Laplacians.
Consider the following subspaces of $C(X)= \{ f \colon X \to \R\}$:
\begin{align*}
    C_c(X) &= \{ f \in C(X) \mid \textup{ the support of } f \textup{ is finite} \} \\
    \ell^p(X,m) &= \{ f\in C(X) \mid \sum_{x \in X} |f(x)|^p m(x) < \infty \} \quad \textup{ for } p \in [1,\infty) \\
    \ell^\infty(X) & = \{ f \in C(X) \mid \sup_{x \in X}|f(x)| <\infty \}.
\end{align*}
For $p \in [1,\infty)$, $\ell^p(X,m)$ is a Banach space
with norm $\|f\|_p^p = \sum_{x \in X}|f(x)|^pm(x)$.
For $p=2$, $\ell^2(X,m)$ is a Hilbert space with inner product
$$\as{f,g} = \sum_{x \in X}f(x)g(x)m(x)$$
and associated norm $\|f\|^2= \as{f,f}$. Finally, $\ell^\infty(X)$ is a Banach
space with norm
$$\|f\|_\infty = \sup_{x \in X}|f(x)|.$$

For $f \in C(X)$, let $\Q= \Q_{b,c}$ be given by
$$\Q(f) = \frac{1}{2}\sum_{x,y \in X}b(x,y)(f(x)-f(y))^2 + \sum_{x \in X}c(x)f^2(x)$$
and let
$$\D = \{ f \in C(X) \mid \Q(f)<\infty\}$$
denote the space of \emph{functions of finite energy}. For $f, g \in \D$,
let
$$\Q(f,g) =\frac{1}{2} \sum_{x,y \in X}b(x,y)(f(x)-f(y))(g(x)-g(y))+\sum_{x \in X}c(x)f(x)g(x)$$
be the \emph{energy form} and $\Q(f)$ the \emph{energy} of $f$.
For a spherically symmetric function $f$, the energy
is given by
$$\Q(f) = \sum_{r=0}^\infty \pt B(r)(f(r+1)-f(r))^2 + \sum_{r=0}^\infty c(S_r)f^2(r).$$

We consider two restrictions of $\Q$. The domain of the first restriction is the closure
of the finitely supported functions in the form norm. More precisely,
let $\QD$ denote the restriction of $\Q$ to
$$D(\QD)= \ov{C_c(X)}^{\|\cdot\|_\Q}$$
where $\| \vp \|^2_\Q= \|\vp\|^2 + \Q(\vp) $
denotes the form norm. For the second form, we let $\QN$ denote the restriction
of $\Q$ to
$$D(\QN)= \D \cap \ell^2(X,m).$$
A graph satisfies \emph{form uniqueness} if $\QD=\QN$.
In this case, there is a unique closed restriction of $\Q$ whose domain
contains the finitely supported functions, see Exercise~1.10 in \cite{KLW21}.

We now reformulate both form uniqueness and weak spherical symmetry
in terms of operators. Let
$$\F = \{ f \in C(X) \mid \sum_{y \in X} b(x,y) |f(y)|<\infty 
\textup{ for all } x \in X\}$$ 
and, for $f \in \F$ and $x \in X$, let
$$\De f(x) = \frac{1}{m(x)} \sum_{y \in X}b(x,y)(f(x)-f(y)) + \frac{c(x)}{m(x)} f(x)$$
denote the \emph{formal Laplacian}. 
For locally finite graphs, note that $\F = C(X)$.
Furthermore, it turns out that the restriction of $\De$ to
$\ell^p(X,m)$ is bounded for some (equivalently, all) $p \in [1,\infty]$
if and only if the weighted vertex degree function $\Deg$ is bounded,
see \cite{HKLW12, KLW21}.

For $\vp, \psi \in C_c(X)$ it is easy to see that
the following Green's formula holds:
$$\Q(\vp, \psi) = \as{\De \vp, \psi}= \as{\vp, \De \psi}$$
giving a connection between the energy form and Laplacian.
In particular, the restriction of the energy form to any
$\ell^p(X,m)$ space is bounded if and only if the restriction
of $\De$ to that $\ell^p(X,m)$ space is bounded. This boundedness is
equivalent to the boundedness of the weighted vertex degree as
noted above.

We now give an equivalent formulation of form uniqueness in terms of operators.
Let $L$ be a positive operator with associated closed form $Q$. Then
$L$ is a \emph{realization} of $\De$ if $L$ is a restriction of $\De$ and if
the domain of $Q$ contains $C_c(X)$. Furthermore, $L$ is a \emph{Markov realization}
of $\De$ if $L$ is a realization of $\De$ and $Q$ is additionally a Dirichlet form, 
i.e., a positive closed form which is compatible with normal contractions,
see \cite{FOT94, KLW21} for more on Dirichlet forms.

A graph satisfies \emph{Markov uniqueness} if there exists a unique Markov realization
of $\De$. The forms $\QD$ and $\QN$ are both Dirichlet forms and the arising operators
are Markov realizations of $\De$, see Example~3.10 in \cite{KLW21}. 
Markov uniqueness is equivalent 
to form uniqueness by results found in \cite{Sch20a, Sch20b}, see also \cite{HKLW12}
for the locally finite case and Theorem~3.12 in \cite{KLW21} for another presentation of this result.

Whenever $\De(C_c(X)) \subseteq \ell^2(X,m)$, we say that
$\De_c=\De \vert_{C_c(X)}$ is \emph{essentially self-adjoint}
if $\De_c$ has a unique self-adjoint extension. In this case, it is 
clear that $\De$ has a unique Markov realization and thus
form uniqueness holds. We will discuss further connections
to this property later.

For us, the most important viewpoint for form uniqueness will be that of $\alpha$-harmonic
functions. A function $u \in \F$ is \emph{$\al$-harmonic} for $\al \in \R$ if
$$(\De + \al)u = 0.$$
Our main technical tool is that the following general statement.
\begin{lemma}[Form uniqueness and $\al$-harmonic functions]\label{l:al-harmonic}
Let $(b,c)$ be a graph over $(X,m)$. Then the following statements are equivalent:
\begin{itemize}
\item[(i)] $\QD \neq \QN.$
\item[(ii)] There exists a non-zero $u \geq 0$ and $\al>0$ 
with $u \in \D \cap \ell^2(X,m)$ such that
$$(\De + \al)u=0.$$
\end{itemize}
\end{lemma}
See Theorem~3.2 in \cite{KLW21} for a proof.
Note that, on a connected graph, if $u \geq 0$ is non-zero and satisfies
$(\De+\al)u =0$, then $u >0$ by a Harnack principle, see Corollary~4.2
in \cite{KLW21}. 
We call functions satisfying $u \geq 0$ \emph{positive} and functions
satisfying $u>0$ \emph{strictly positive}. Thus, on connected graphs,
we may always  work with strictly positive $\al$-harmonic functions
when investigating form uniqueness.

Furthermore, note that, for any vertex $x_0 \in X$
such that $u(x_0)>0$, the equation $(\De+\al)u(x_0) =0$ for $\al>0$ implies
that $x_0$ must have a neighbor $x_1 \sim x_0$ such that $u(x_1)>u(x_0)$.
This basic observation has been used to obtain criteria for 
form uniqueness previously and 
will be used repeatedly in what follows. For example, it yields triviality
of $\alpha$-harmonic functions in $\ell^2(X,m)$ whenever the vertex measure
does not decay too strongly. More specifically, if every infinite path
has infinite measure or, more strongly, if the measure is uniformly bounded
from below, then essential self-adjointness and, thus, 
Markov uniqueness follow easily, see \cite{Woj08, KL12, KLW21}.

\subsection{Averaging operator}
Consider now locally finite connected graphs with $O \subseteq X$ finite and let $S_r = S_r(O)$. 
Define the 
\emph{formal averaging operator} with respect to $O$ as 
$$\A f(x) = \frac{1}{m(S_r)} \sum_{y \in S_r}f(y)m(y)$$
for $x \in S_r$ and $f \in C(X)$.
It is clear that $\A f$ is a spherically symmetric function
for any $f \in C(X)$.
We first show that restricting $\A$ to any $\ell^p(X,m)$
space produces a bounded operator. For $p=2$, this was observed
in \cite{KLW13}.
%and denote the restriction of $\A$ to $\ell^2(X,m)$ by $A$. It is not hard to see
%that $A$ is a bounded operator, see Lemma~9.7 in \cite{KLW21}.
\begin{lemma}[$A$ is bounded]\label{l:average}
Let $(b,c)$ be a locally finite connected graph with $O \subseteq X$ finite.
If $\A$ is the averaging operator with respect to $O$, then
the restriction of $\A$ to $\ell^p(X,m)$ for $p \in [1,\infty]$
gives a bounded operator with norm $1$.
\end{lemma}
\begin{proof}
Denote the restriction of $\A$
to $\ell^p(X,m)$ by $A^{(p)}$. We start with the case of $p=1$. Let $f\in \ell^1(X,m)$ 
and calculate as follows:
\begin{align*}
\|A^{(1)}f\|_1 &= \sum_{x \in X}|\A f(x)| m(x) = \sum_{r=0}^\infty 
\sum_{x \in S_r} |\A f(x)| m(x) \\
&= \sum_{r=0}^\infty \sum_{x \in S_r} \left| \frac{1}{m(S_r)} \sum_{y \in S_r} f(y) m(y)\right| m(x) \\
&\leq \sum_{r=0}^\infty  \sum_{y \in S_r} |f(y)| m(y) = \|f\|_1.
\end{align*}
This shows $\|A^{(1)}\| \leq 1$ and since $\A f=f$ for any spherically
symmetric function, $\|A^{(1)}\|=1$.

Consider now the case of $p=\infty$.
For $f \in \ell^\infty(X)$ and $x \in S_r$ for $r \in \N_0$
$$|\A f(x)| = \left| \frac{1}{m(S_r)} \sum_{y \in S_r}f(y)m(y) \right| 
\leq \|f\|_\infty$$
which implies $\|A^{(\infty)} f\|_\infty \leq \|f\|_\infty$ and thus, as for $p=1$,
$\A$ gives a bounded operator on $\ell^\infty(X)$ with operator
norm $1$. The statement for general $p \in [1,\infty]$ now follows
by the Riesz-Thorin interpolation theorem, see \cite{SW71}.
\end{proof}

It tuns out that a locally finite connected graph $(b,c)$ over $(X,m)$ is weakly spherically symmetric if and only if
the formal Laplacian and averaging operators commute,
i.e.,
$$\De \A = \A \De$$
on $C(X) =\F$,  see Lemma~3.3 in \cite{KLW13} or Lemma~9.8 in \cite{KLW21}
for a proof of this equivalence.
This gives an operator-theoretic viewpoint
on weak spherical symmetry.

For the convenience of the reader, we now state some basic
facts about weakly spherically symmetric graphs and the
averaging operator which will be used in what
follows.
\begin{lemma}[Basic facts about weakly spherically symmetric graphs]\label{l:basic_facts_wss}
Let $(b,c)$ be a weakly spherically symmetric graph over $(X,m)$
with formal averaging operator $\A$.
\begin{itemize}
\item[(a)] For $r \in \N_0$,
$$\pt B(r)= \ka_+(r)m(S_r) = \ka_-(r+1) m(S_{r+1})$$
and 
$$q(r)m(S_r)=c(S_r).$$
\item[(b)] If $f \in \F$ is spherically symmetric and $r \in \N_0$,
then $\De f$ is spherically symmetric with
$$\Delta f(r) = \ka_+(r)(f(r)-f(r+1)) + \ka_-(r)(f(r)-f(r-1)) + q(r)f(r)$$
or, equivalently,
$$\Delta f(r)m(S_r)= \pt B(r)(f(r)-f(r+1)) + \pt B(r-1)(f(r)-f(r-1))+c(S_r)f(r).$$
%\item[(c)] If $f \in \D$ is spherically symmetric, then
%$$\Q(f) = \sum_{r=0}^\infty \pt B(r)(f(r)-f(r+1))^2 + \sum_{r=0}^\infty c(S_r)f^2(r).$$
%\item[(d)] The restriction of $\A$ to $\ell^p(X,m)$ for $p \in [1,\infty]$
%gives a bounded operator with norm $1$.
\item[(c)] If $v \in \F$ is $\al$-harmonic, then $u=\A v$ is spherically
symmetric and $\al$-harmonic.
\end{itemize}
\end{lemma}
\begin{proof}
Property (a) follows from definitions and an interchange
of sums while (b) follows by definitions and by (a).

Statement (c) follows since weak spherical symmetry is equivalent
to the fact that $\De$ and $\A$ commute on $C(X)=\F$, see Lemma~9.8
in \cite{KLW21}. Thus, if $v \in C(X)$ satisfies $(\De +\al)v=0$,
then $u=\A v$ is spherically symmetric and satisfies
$$(\De + \al)u = (\De + \al)\A v = \A (\De+\al)v = 0.$$
This completes the proof.
\end{proof}

\section{Form uniqueness}\label{s:form_uniqueness}
In this section we prove our characterizations for form
uniqueness on weakly spherically symmetric graphs.
We first prove a characterization
in terms of the graph structure. Here, the key steps are to show that energy
is well-behaved with respect to averaging and to study when $\al$-harmonic
functions have finite energy. We then use this result to give a characterization via capacity as well as studying the Feller property for the Neumann semigroup. Finally, we compare our criteria to those
for other properties in the case of no killing term.

\subsection{Form uniqueness via the graph structure}
In this subsection we will prove the characterization of form uniqueness
via the graph structure. For this, there are two essential ingredients. The first
allows us to compare the energy of an averaged function to the original function.
The second gives a characterization for when spherically symmetric $\al$-harmonic
functions have finite energy. We start with the first ingredient.

\begin{lemma}[Averaging and energy]\label{l:averaging_energy}
Let $(b,c)$ be a weakly spherically symmetric graph over $(X,m)$. If $v \in C(X)$, then
$$\Q(\A v) \leq \Q(v).$$
In particular, if $v \in \D$, then $\A v \in \D$.
\end{lemma}
\begin{proof}
We write $\Q_{b,c}=\Q_{b,0} + \Q_{0,c}$ and work with each piece separately. If $v \not \in \D$,
there is nothing to show. Thus, we may assume that $v \in \D$.

We first show the result for $\Q_b = \Q_{b,0}$. Here we use the identity
$$\pt B(r)= \ka_+(r)m(S_r) = \ka_-(r+1) m(S_{r+1})$$
from Lemma~\ref{l:basic_facts_wss}~(a) as well as the following variant on Cauchy--Schwarz called Sedrakyan's inequality
which holds for $\al_i,\be_i \in \R$ with $\be_i > 0$:
$$\frac{\left(\sum_{i} \al_i \right)^2}{\sum_{i} \be_i} \leq \sum_{i} \frac{\al_i^2}{\be_i}.$$
We calculate as follows:
\begin{align*}
\Q_b(\A v) &= \sum_{r=0}^\infty \pt B(r)(\A v(r) - \A v(r+1))^2 \\
&= \sum_{r=0}^\infty \pt B(r)\left(\frac{1}{m(S_r)}\sum_{x \in S_r} v(x)m(x) - \frac{1}{m(S_{r+1})} \sum_{y \in S_{r+1}} v(y)m(y) \right)^2 \\
&= \sum_{r=0}^\infty \pt B(r)\left(\frac{\ka_+(r)}{\pt B(r)}\sum_{x \in S_r} v(x)m(x) - \frac{\ka_-(r+1)}{\pt B(r)} \sum_{y \in S_{r+1}} v(y)m(y) \right)^2 \\
&= \sum_{r=0}^\infty \frac{1}{\pt B(r)}\left(\sum_{x \in S_r}\sum_{y \in S_{r+1}} b(x,y)v(x) - \sum_{y \in S_{r+1}} \sum_{x \in S_r} b(x,y) v(y) \right)^2 \\
&= \sum_{r=0}^\infty \frac{1}{\sum_{x \in S_r}\sum_{y \in S_{r+1}}b(x,y)}\left(\sum_{x \in S_r}\sum_{y \in S_{r+1}} b(x,y)(v(x) - v(y)) \right)^2 \\
&\leq \sum_{r=0}^\infty \sum_{x \in S_r}\sum_{y \in S_{r+1}} b(x,y)(v(x) - v(y))^2 \\
&\leq \Q_b(v)
\end{align*}
which shows the result for $\Q_b$.

For $\Q_c= \Q_{0,c}$, this follows by Cauchy--Schwarz directly as
\begin{align*}
\Q_c(\A v) &= \sum_{x \in X}\frac{c(x)}{m(x)} (\A v(x))^2 m(x) \\
&= \sum_{r=0}^\infty q(r) (\A v(r))^2 m(S_r) \\
&= \sum_{r=0}^\infty q(r) \left(\frac{1}{m(S_r)}\sum_{x \in S_r}v(x)m(x)\right)^2 m(S_r) \\
&\leq \sum_{r=0}^\infty q(r) \sum_{x \in S_r} v^2(x)m(x) \\
&= \sum_{r=0}^\infty \sum_{x \in S_r} c(x)v^2(x) = \Q_c(v).
\end{align*}
Combining the two inequalities completes the proof.
\end{proof}

The second ingredient
is a characterization for when spherically symmetric
$\al$-harmonic functions have finite energy.
We first recall a recurrence formula which 
gives the existence of such functions
as well as criterion for them to be bounded.
\begin{lemma}[Recurrence relation]\label{l:recurrence}
Let $(b,c)$ be a weakly spherically symmetric graph over $(X,m)$.
Then $u$ is spherically symmetric and satisfies $(\De + \al) u =0$
for $\al \in \R$ if and only if
$$u(r+1)-u(r) = \frac{1}{\pt B(r)} \sum_{k=0}^r (q(k) + \al)m(S_k)u(k).$$
In particular, for $\al>0$, $u$ is positive and non-zero if and only if $u(0)>0$
in which case $u$ is monotonically increasing. Analogous statements
hold for the case of $u(0)<0$.

Furthermore, $u \in \ell^\infty(X)$ if and only if 
$$\sum_{r=0}^\infty \frac{(c+m)(B_r)}{\pt B(r)} < \infty. $$

\end{lemma}
\begin{proof}
The recurrence relation follows by induction, see Lemma~9.16 in \cite{KLW21}. 
The ``in particular'' statements are easy consequences of the recurrence relation.
Finally, the boundedness criterion follows from the recurrence formula
and some basic estimates, see Lemma~9.26 in \cite{KLW21}.
\end{proof}

We now give a characterization for when $\al$-harmonic spherically
symmetric functions have finite energy.

\begin{lemma}[Spherically symmetric energy]\label{l:harmonic_energy}
Let $(b,c)$ be a weakly spherically symmetric graph over $(X,m)$.
Let $u$ be spherically symmetric, non-zero and satisfying $(\De + \al) u =0$
for $\al>0$. Then $u \in \D$ if and only if
$$c(X)<\infty
\qquad and \qquad \sum_{r=0}^\infty \frac{\left( m(B_r) \right)^2}{\pt B(r)} < \infty.$$
\end{lemma}
\begin{proof}
As pointed out previously, the energy of 
a spherically symmetric function $u$ is given by
\begin{equation}\label{e:energy}
\Q(u) = \sum_{r=0}^\infty \pt B(r)(u(r+1)-u(r))^2 + \sum_{r=0}^\infty c(S_r)u^2(r).
\end{equation}

We now estimate the first summand in \eqref{e:energy}, i.e., the energy
coming from the edge boundary of the ball.
From the recurrence relation in Lemma~\ref{l:recurrence} we get
$$\pt B(r)(u(r+1)-u(r))^2 = \frac{1}{\pt B(r)} \left(\sum_{k=0}^r (q(k) + \al)m(S_k)u(k)\right)^2.$$
Assuming that $u(0)>0$ so that $u$ is strictly increasing by Lemma~\ref{l:recurrence}, we
estimate this from above and from below as follows
\begin{align*}
\frac{1}{\pt B(r)}\left(\sum_{k=0}^r (q(k) + \al)m(S_k)\right)^2 u^2(0)  &\leq  \pt B(r)(u(r+1)-u(r))^2 \\
 &\leq \frac{1}{\pt B(r)}\left(\sum_{k=0}^r (q(k) + \al)m(S_k)\right)^2 u^2(r).
\end{align*}
Equivalently,
\begin{equation}\label{e:inequalities}
\frac{\left((c + \al m)(B_r)\right)^2}{\pt B(r)}  u^2(0) \leq \pt B(r)(u(r+1)-u(r))^2
\leq \frac{\left((c + \al m)(B_r)\right)^2}{\pt B(r)} u^2(r).
\end{equation}

If $c(X)<\infty$ and 
$\sum_r \frac{\left(m(B_r) \right)^2}{\pt B(r)} < \infty$, then 
$\sum_r \frac{((c + m)(B_r))^2}{\pt B(r)} < \infty$.
This implies $\sum_r \frac{(c + m)(B_r)}{\pt B(r)} < \infty$
and, thus, $u\in \ell^\infty(X)$ by Lemma~\ref{l:recurrence}. 
Therefore, it follows from the upper bound in \eqref{e:inequalities} that 
$$\pt B(r)(u(r+1)-u(r))^2
\leq \frac{\left((c + \al m)(B_r)\right)^2}{\pt B(r)} \|u\|_\infty^2$$
for all $r \in \N_0$ and thus $\sum_r \pt B(r)(u(r+1)-u(r))^2< \infty$.
Furthermore, $\sum_r c(S_r)u^2(r) \leq c(X)\|u\|_\infty^2$
so that $ u \in \D$ by \eqref{e:energy}.

Conversely, if $c(X)=\infty$, then $\sum_r c(S_r)u^2(r)=\infty$ since $u$ is strictly increasing  and thus
$u \not \in \D$ by \eqref{e:energy}.
On the other hand,
if $\sum_r \frac{\left(m(B_r) \right)^2}{\pt B(r)} = \infty$, then
$\sum_r \frac{\left((c + m)(B_r) \right)^2}{\pt B(r)} = \infty$ and thus
$\sum_r \pt B(r)(u(r+1)-u(r))^2= \infty$ from the lower bound in \eqref{e:inequalities}.
Therefore, $u \not \in \D$ by \eqref{e:energy}. 

The same argument works if $u(0)<0$. This completes the proof.
\end{proof}

As a direct consequence, 
we highlight that spherically symmetric $\al$-harmonic functions of finite energy 
are bounded. This was already used in the proof above.

\begin{corollary}[Finite energy implies bounded]\label{c:bounded}
Let $(b,c)$ be a weakly spherically symmetric graph over $(X,m)$.
Let $u$ be spherically symmetric, non-zero and satisfying $(\De + \al) u =0$
for $\al>0$. If $u \in \D$, then $u \in \ell^\infty(X).$
\end{corollary}
\begin{proof}
This follows directly by combining Lemmas~\ref{l:recurrence}~and~\ref{l:harmonic_energy}.
\end{proof}

We now put the pieces together to prove our first major result.
\begin{theorem}[Form uniqueness and graph structure]\label{t:form_uniqueness_weights}
Let $(b,c)$ be a weakly spherically symmetric graph over $(X,m)$.
Then the following statements are equivalent:
\begin{itemize}
\item[(i)] $\QD \neq \QN.$
\item[(ii)] 
$(c+m)(X) < \infty$ \qquad \textup{and} \qquad
$\displaystyle{\sum_{r=0}^\infty \frac{1}{\pt B(r)} < \infty}.$
\end{itemize}
\end{theorem}
\begin{proof}
(i) $\Longrightarrow$ (ii): Since $\QD \neq \QN$, there exists a positive non-trivial
$\al$-harmonic function $v \in \D \cap \ell^2(X,m)$ for $\al>0$ 
by Lemma~\ref{l:al-harmonic}. 

Let $u =A v$ where $A=A^{(2)}$ is the restriction of the averaging operator $\A$ to $\ell^2(X,m)$. 
By Lemma~\ref{l:average}, $A$ is bounded
on $\ell^2(X,m)$ and thus $u \in \ell^2(X,m)$.
As $\A$ commutes with $\De$ by spherical symmetry, it follows that $u$ is $\al$-harmonic by Lemma~\ref{l:basic_facts_wss}~(c).
As $u$ is spherically symmetric, positive and non-trivial, it follows that
$u$ is increasing by Lemma~\ref{l:recurrence} and
since $u \in \ell^2(X,m)$, $m(X)<\infty$ now follows.

By Lemma~\ref{l:averaging_energy}, $u \in \D$. Hence,
by Lemma~\ref{l:harmonic_energy}, we now obtain
$$c(X)<\infty 
\qquad \textup{ and } \qquad \sum_{r=0}^\infty \frac{\left(m(B_r) \right)^2}{\pt B(r)} < \infty$$
and, since we have already shown $m(X)<\infty$, it now follows that
$$(c+m)(X)<\infty 
\qquad \textup{ and } \qquad \sum_{r=0}^\infty \frac{1}{\pt B(r)} < \infty $$
as claimed.

(ii) $\Longrightarrow$ (i): 
Let $u \in C(X)$ be spherically symmetric with $u(0)>0$
and satisfying $(\Delta + \al)u=0$ for $\al>0$. Such a function exists 
by Lemma~\ref{l:recurrence}.
The assumptions $(c+m)(X) < \infty$ and 
$\sum_{r=0}^\infty \frac{1}{\pt B(r)} < \infty$
together imply
$$c(X)<\infty
\qquad \textup{and} \qquad \sum_{r=0}^\infty \frac{\left( m(B_r) \right)^2}{\pt B(r)} < \infty.$$
Therefore, $u \in \D$ by Lemma~\ref{l:harmonic_energy}. Furthermore, $u \in \ell^\infty(X)$
by Corollary~\ref{c:bounded}.
As $m(X)<\infty$, this implies $u \in \D \cap \ell^2(X,m)$ and thus $\QD \neq \QN$
by Lemma~\ref{l:al-harmonic}. This completes the proof.
\end{proof}

\begin{remark}
For the case of $c=0$, the implication (ii) $\Longrightarrow$ (i)
in the above is known by a result found in \cite{GHKLW15, Sch17} along 
with a characterization for transience of weakly spherically symmetric
graphs. More specifically,
in the case of graphs where $c=0$ and $m(X)<\infty$, 
$\QD \neq \QN$ is equivalent to the graph being transient or
stochastically incomplete. In the case of weakly
spherically symmetric graphs with $c=0$, this
is equivalent to $\sum_r 1/\pt B(r)<\infty,$ see
further discussion in Subsection~\ref{ss:connections} below.

Furthermore, if $m(X)<\infty$, by the 
equivalence between $\QD \neq \QN$, transience and stochastic
incompleteness, it is possible
to establish comparison theorems between weakly spherically symmetric
graphs and general locally finite graphs 
as in Theorem~6 of \cite{KLW13}, see also
Theorems~9.24~and~9.30 in \cite{KLW21}. More specifically, for $m(X)<\infty$, if a connected
locally finite graph has stronger degree growth than a weakly spherically symmetric
graph which is not form unique, then the general graph is not form unique.
Similarly, if a locally
finite graph has weaker degree growth than a weakly spherically symmetric
graph which is form unique, then the general graph is form unique.
See also, \cite{Woj17} for analogous statements for the Feller property.
For the case of $m(X)=\infty$, it is not clear how to establish
such comparison results. For an analogous proof to the results mentioned above,
we would need to be able
to use $\al$-sub/super harmonic functions in the criterion 
for form uniqueness in Lemma~\ref{l:al-harmonic}.
\end{remark}

\subsection{Form uniqueness via capacity}\label{ss:capacity}
We now give a characterization of form uniqueness
for weakly spherically symmetric graphs via capacity
of the Cauchy boundary. 
For general locally finite graphs with $c=0$, the
connection between form uniqueness and capacity was studied
in \cite{HKMW13}. 
For more general operators, which allow
for a killing term, and a generalized notion of capacity which involves
$1$-excessive functions, see \cite{Sch20b}.
In order to make our presentation more self-contained, we extend
the required results from \cite{HKMW13} to the case
of general $c$ in Appendix~\ref{s:appendix}.

In this subsection we assume that all graphs are locally finite and connected.
We start by recalling some definitions.
Let $\Ga_{x,y}$ denote the
set of paths starting at $x$ and ending at $y$ for $x,y \in X$. Let
$\si \colon X\times X \to [0,\infty)$ be symmetric such that $\si(x,y)>0$ if and only if
$x \sim y$ and call $\si(x,y)$ the \emph{length} of an edge $x \sim y$. For
$(x_k) \in \Ga_{x,y}$, let
$$l_\si((x_k))= \sum_k \si(x_k, x_{k+1})$$
denote the \emph{length} of the path $(x_k)$. For $x,y \in X$
with $x \neq y$, let
$$d_\si(x,y) = \inf_{(x_k)\in \Ga_{x,y}} l_\si((x_k))$$
with $d_\si(x,x)=0$, and call $d_\si$ a \emph{path metric}.
Note that, by local finiteness, $d_\si$ is a metric for every $\si$, not
just a pseudometric, see Lemma~11.8 in \cite{KLW21}.
Let $\ov{X}^\si$ denote the metric completion of $X$ with respect to $d_\si$ and
let 
$$\pt_\si X = \ov{X}^\si \setminus X$$
denote the \emph{Cauchy boundary} of the graph.
For locally finite graphs and path metrics, note that points in the
Cauchy boundary may be identified with geodesics of finite length,
see Appendix~A~ in \cite{HKMW13} or Section~11.2 in \cite{KLW21} for
more details on this viewpoint.

A pseudometric $\varrho\colon X \times X \to [0,\infty)$ 
is said to be \emph{intrinsic} if
$$\sum_{y \in X}b(x,y)\varrho^2(x,y) \leq m(x)$$
for all $x \in X$
and is said to be a \emph{path metric} if $\varrho=d_\si$ for
some length function $\si$. Furthermore,
a path metric $d_\si$ is called \emph{strongly intrinsic} if
$$\sum_{y \in X}b(x,y) \si^2(x,y) \leq m(x)$$
for all $x \in X$. 

\begin{example}\label{ex:intrinsic} We give two examples to illustrate the
notion of an intrinsic path metric.

(1) \textbf{Combinatorial graph distance:}
If $\si(x,y)=1$ for every $x\sim y$, then we recover the 
combinatorial graph distance $d$.
In this case, $\pt_\si X =\emptyset$
and $d$ is equivalent to a strongly intrinsic metric if and only if
the weighted vertex degree without the killing term is bounded.
As all forms and operators are bounded in this case, form uniqueness follows easily.

(2) \textbf{Degree path metric:}
Letting
$$\si(x,y) = \min\left\{\Deg(x)^{-1/2},\Deg(y)^{-1/2}\right\}$$
for $x \sim y$, where $\Deg$ is the weighted
vertex degree, gives a strongly intrinsic path metric
as can be checked by a direct
calculation. See \cite{Hua11b, FLW14} for
early appearances and \cite{Kel15, KLW21} for more discussion.
\end{example}

\begin{remark}[On the use of intrinsic metrics]
Intrinsic metrics are used to control the energy of cut-off
functions defined via the metric, see Appendix~\ref{s:appendix}
for an example of this. They are especially useful in situations
where the vertex degree is unbounded and the combinatorial graph distance
does not give results parallel to those for manifolds or strongly local
Dirichlet forms.
See \cite{FLW14} for general theory, \cite{Stu94} for strongly local Dirichlet
forms, and Part~3 in \cite{KLW21} for a variety
of uses for graphs.

For us, the most relevant use is that metric completeness, i.e., 
when $\pt_\si X = \emptyset$,
for locally finite graphs and intrinsic path metrics
implies that $\De_c=\De \vert_{C_c(X)}$ is essentially self-adjoint
which then implies form uniqueness. In fact, the statement is more general, in that
if the weighted vertex degree is bounded on distance balls
with respect to any intrinsic pseudometric, then $\De_c$
is essentially self-adjoint, see Theorem~1 in \cite{HKMW13}. 
For locally finite graphs and path metrics, metric completeness
translates into finiteness of distance balls, giving 
a Hopf-Rinow type theorem, see \cite{HKMW13,
KM19, KLW21}.
Thus, combining these results gives an analogue to a theorem of Gaffney on Riemannian
manifolds stating that geodesic completeness
implies Markov uniqueness \cite{Gaf51, Gaf54}, see also \cite{Che73, Str83}.
However, note that, unlike in the case of manifolds or strongly
local Dirichlet forms, there is no maximal intrinsic metric for graphs.
This was already observed in \cite{FLW14}, see also further developments
in \cite{LSS24}.

The notion of a strongly intrinsic path metric is required
as we want to study graphs that remain following the removal
of a neighborhood of the Cauchy boundary. 
For such graphs it is important that path metrics remain intrinsic
following the removal of some edges.
For details, see \cite{HKMW13} and also 
Appendix~\ref{s:appendix}.
\end{remark}

We now recall a notion of size for a set that is based
on both the measure and the energy.
For a set $K \subseteq X$, let
$$\cp(K) = \inf\{ \|u\|_\Q \mid u \in D(\QN) \textup{ with } u \geq 1 \textup{ on } K\}$$
where $\cp(K)=\infty$ if there is no such function $u$.
Note that $\cp(\emptyset)=0$ by letting $u=0$ in the definition.
For $K \subseteq \ov{X}^\si$, let
$$\cp(K) = \inf\{ \cp(U \cap X) \mid K \subseteq U \textup{ with } U \subseteq \ov{X}^\si 
\textup{ open} \}.$$

Given these definitions, we now state the following general result from \cite{HKMW13}.
\begin{theorem}[Theorem~3 in \cite{HKMW13}]\label{t:capacity_general}
Let $(b,c)$ be a locally finite graph over $(X,m)$. Let $d_\si$
be a strongly intrinsic path metric with Cauchy
boundary $\pt_\si X$. If $\cp(\pt_\si X)<\infty$, then
the following statements are equivalent:
\begin{itemize}
\item[(i)] $\QD \neq \QN$.
\item[(ii)] $0< \cp(\pt_\si X)$.
\end{itemize}
\end{theorem}
\begin{remark}
Note that \cite{HKMW13} only allows for $c=0$, however, the arguments
carry over easily to the case of a non-zero killing term, see
Appendix~\ref{s:appendix} for details.
See also Theorem~6.9 in \cite{Sch20b} for an extension using a notion
of capacity which uses 1-excessive functions. In order to obtain Theorem~\ref{t:capacity_general}
as a special case, the $1$-excessive
function is the constant function $h=1$. In \cite{Sch20b}, the assumption
that $\cp(\pt_\si X)<\infty$ is reflected in the requirement that $1 \in D(\QN)$
which is equivalent to $(c+m)(X)<\infty$ and implies $\cp(\pt_\si X)<\infty$.
For related results for manifolds, see \cite{GM13, Mas99, Mas05}.
\end{remark}

Note that if $\cp(\pt_\si X) =\infty$, then a graph may or may not 
be form unique, see Examples~5.2~and~5.4 in \cite{HKMW13}.
For weakly spherically symmetric graphs, we now sharpen the result above by
showing that
infinite capacity of the boundary implies form uniqueness.

\begin{theorem}[Form uniqueness and capacity]\label{t:capacity}
Let $(b,c)$ be a weakly spherically symmetric graph over $(X,m)$. Let $d_\si$
be a strongly intrinsic path metric with Cauchy
boundary $\pt_\si X$. Then the following statements are equivalent:
\begin{itemize}
\item[(i)] $\QD \neq \QN$.
\item[(ii)] $0< \cp(\pt_\si X) <\infty$.
\end{itemize}
\end{theorem}
\begin{proof}
(i) $\Longrightarrow$ (ii): 
We argue by contraposition.
If $\cp(\pt_\si X)=0$, then the graph satisfies form uniqueness
by Theorem~\ref{t:capacity_general} directly above.
If $\cp(\pt_\si X) =\infty$ and $(c+m)(X)<\infty$, then $1 \in \D \cap \ell^2(X,m)$
where $1$ is the constant function which is $1$ on every vertex.
This would imply $\cp(\pt_\si X) \leq (c+m)(X)<\infty$, giving a contradiction. 
Hence, $(c+m)(X)=\infty$, which gives $\QD=\QN$ by Theorem~\ref{t:form_uniqueness_weights}.

(ii) $\Longrightarrow$ (i): This is part of Theorem~\ref{t:capacity_general}, 
see Lemma~\ref{l:capacity_alternative} in Appendix~\ref{s:appendix} for a proof.
\end{proof}

\subsection{Form uniqueness and the Feller property for the Neumann semigroup}
In this subsection we characterize the Feller property for the Neumann semigroup extending a result from \cite{KMW25}
for birth--death chains. For this, we let $c=0$ throughout.

We first introduce the relevant concepts. Denote the semigroups associated
to $\QD$ and $\QN$ by $\PD$ and $\PN$ for $t \geq 0$, respectively, 
and call them the
\emph{Dirichlet} and \emph{Neumann semigroups}. $\PN$ satisfies the 
\emph{Feller property} or is \emph{Feller} if
$$\PN(C_c(X)) \subseteq C_0(X)$$
where $C_0(X) = \ov{C_c(X)}^{\|\cdot\|_\infty}$
denotes the space of functions vanishing at infinity, with an analogous
definition for $\PD$. 

For weakly spherically symmetric graphs, a characterization for $\PD$ to be
Feller is given in \cite{Woj17} in analogy to the case of spherically
symmetric manifolds presented in \cite{PS12}. 
Furthermore, it turns out that $\PN$ is 
Feller if and only if $\PD$ is Feller and $\PD = \PN$, see Corollary~5.2 in \cite{KMW25}.
Given these two ingredients along with Theorem~\ref{t:form_uniqueness_weights}, we now present a characterization of
the Feller property for the Neumann semigroup on weakly spherically symmetric
graphs.
Let $B_r^c$ denote the complement of the ball of radius $r$ about $O \subseteq X$, i.e.,
$$B_r^c = \{ x \in X \mid d(x,O) >r \}.$$

\begin{theorem}[Feller property of the Neumann semigroup]\label{t:Feller}
Let $b$ be a weakly spherically symmetric graph over $(X,m)$.
Then the following statements are equivalent:
\begin{itemize}
\item[(i)] $\PN$ is not Feller.
\item[(ii)] $\displaystyle{\sum_{r=0}^\infty \frac{m(B_r^c)}{\pt B(r)} < \infty}.$
\item[(iii)] There exists a non-zero $u \geq 0$ and $\al>0$ 
with $u \in \ell^1(X,m)$ such that
$$(\De + \al)u=0.$$ 
\end{itemize}
\end{theorem}
\begin{proof}
(i) $\Longleftrightarrow$ (ii):
Corollary~5.2 in \cite{KMW25} gives that $\PN$ is not Feller if and only
if $\PD$ is not Feller or $\QD \neq \QN$. Furthermore, Theorem~4.13 in \cite{Woj17}
gives that $\PD$ is not Feller if and only if
$\sum_r \frac{1}{\pt B(r)}=\infty$ and $\sum_r \frac{m(B_r^c)}{\pt B(r)}<\infty$.
Combining these results with Theorem~\ref{t:form_uniqueness_weights}
gives the equivalence between (i) and (ii).

(i) $\Longrightarrow$ (iii): This is given by Theorem~6.1 in \cite{KMW25} and does
not require weak spherical symmetry.

(iii) $\Longrightarrow$ (ii): By Lemma~\ref{l:average}, the restriction of $\A$
to $\ell^1(X,m)$ gives a bounded operator and preserves 
$\alpha$-harmonicity by Lemma~\ref{l:basic_facts_wss}~(c). 
Thus, we may assume that the positive $\al$-harmonic function
$u \in \ell^1(X,m)$ is spherically symmetric. As $u$ is non-zero, it follows
from Lemma~\ref{l:recurrence} that $u(0)>0$ and that $u$ is strictly increasing.
As $u \in \ell^1(X,m)$, it follows that $m(X)<\infty$.

Now, from $(\De + \al)u = 0$, we easily obtain
$$\pt B(r)(u(r+1)-u(r)) \geq \pt B(r-1) (u(r)-u(r-1))$$
see Lemma~\ref{l:basic_facts_wss}~(b) for the needed formula.
Iterating this down to $0$ gives
\begin{align*}
\pt B(r)(u(r+1)-u(r)) &\geq \pt B(0) (u(1)-u(0)) \\
&= -\De u(0)m(O) = \al u(0)m(O) = C >0
\end{align*}
so that $\pt B(r)(u(r+1)-u(r)) \geq C >0$
for all $r \in \N_0$. Hence,
$$u(r+1) \geq  u(r)+ \frac{C}{\pt B(r)} \geq C \sum_{k=0}^r \frac{1}{\pt B(k)}.$$
This yields
\begin{align*}
    \sum_{r=0}^\infty u(r)m(S_r) &\geq C \sum_{r=0}^\infty \left(\sum_{k=0}^{r-1} \frac{1}{\pt B(k)}\right) m(S_r) \\
    &= C \sum_{r=0}^\infty \frac{\sum_{k=r+1}^\infty m(S_k)}{\pt B(r)} =  C \sum_{r=0}^\infty \frac{m(B_r^C)}{\pt B(r)}.
\end{align*}
As $u \in \ell^1(X,m)$, this implies $\sum_{r=0}^\infty \frac{m(B_r^c)}{\pt B(r)} < \infty$ which completes the proof.
\end{proof}

\begin{remark}
The result above extends Theorem~8.2 in \cite{KMW25} from the case
of birth--death chains to all weakly spherically symmetric graphs.
Furthermore, note that (iii) does not imply (i) for general graphs,
see Example~6.5 in \cite{KMW25}.
\end{remark}

\subsection{Connections with other properties}\label{ss:connections}
In this subsection, we compare our characterizations of form uniqueness (Theorem~\ref{t:form_uniqueness_weights})
and the Feller property for the Neumann semigroup (Theorem~\ref{t:Feller})
to characterizations for other properties on weakly spherically symmetric graphs with no
killing term, i.e., when $c=0$.
 
We first define the other properties.
Recall that a connected graph $b$ over $(X,m)$
is called \emph{transient} if
$$\int^\infty \PD 1_x \ dt < \infty$$
for some (all) $x \in X$, where $1_x$ denotes the characteristic function of 
$\{x\}$. Furthermore, a connected graph is 
called \emph{stochastically incomplete} if 
$$\PD 1 < 1$$
where $1$ denotes the function that takes the value $1$ on every vertex and the
semigroup defined on $\ell^2(X,m)$ is extended to $\ell^\infty(X)$ via monotone 
limits. 

Given these definitions, we then have the following characterizations
for weakly spherically symmetric graphs with $c=0$:
\begin{align*}
\QD \neq \QN \qquad & \Longleftrightarrow \qquad  m(X) < \infty \quad \textup{ and } \quad \sum_{r=0}^\infty \frac{1}{\pt B(r)} < \infty. \\
\PD 1 < 1 \qquad & \Longleftrightarrow \qquad \sum_{r=0}^\infty \frac{m(B_r)}{\pt B(r)} < \infty. \\
\int^\infty \PD 1_x \ dt < \infty \qquad & \Longleftrightarrow \qquad \sum_{r=0}^\infty \frac{1}{\pt B(r)} < \infty. \\
%\Delta_c \textup{ is not essentially self-adjoint} \qquad & \Longleftrightarrow \qquad
%\textup{ the graph is a birth--death chain and } \\
%& \qquad \qquad \qquad \sum_{r=0}^\infty \left(\sum_{k=0}^r \frac{1}{b(k,k+1)} \right)^2 m(r+1) < \infty \\
\PD \textup{ is non-Feller} \qquad & \Longleftrightarrow \qquad \sum_{r=0}^\infty \frac{1}{\pt B(r)} = \infty \quad \textup{ and } \quad \sum_{r=0}^\infty \frac{m(B_r^c)}{\pt B(r)} < \infty.  \\
\PN \textup{ is non-Feller} \qquad & \Longleftrightarrow \qquad \sum_{r=0}^\infty \frac{m(B_r^c)}{\pt B(r)} < \infty. 
\end{align*}

For the result on stochastic incompleteness, see \cite{KLW13}, for transience, see 
\cite{Woe09, KLW21},
for the failure of the Feller property for the Dirichlet semigroup, see \cite{Woj17}.
In particular, we see that $\QD \neq \QN$ implies stochastic incompleteness, transience
and the failure of the Feller property for the Neumann semigroup. This has already
been observed in \cite{HKLW12, KMW25} and does not require the weak spherical
symmetry. Furthermore, as previously noted, when $m(X)<\infty$, the properties
of non-form uniqueness, stochastic incompleteness, and transience are all equivalent
for general graphs with $c=0$, see \cite{Sch17, GHKLW15}.

Failure of form uniqueness also implies the failure of essential self-adjointness
for general graphs as follows by general principles, e.g., \cite{HKLW12}. 
For birth--death chains, a result of Hamburger states that $\Delta_c$ is 
essentially self-adjoint if and only if
$$\sum_{r=0}^\infty \left( \sum_{k=0}^r \frac{1}{b(k,k+1)}\right)^2 m(r+1)<\infty$$
see \cite{Ham20a, Ham20b} for the original work of Hamburger and \cite{IKMW25} for
a recent exposition. However, it is not clear that this criterion extends 
to all weakly spherically symmetric graphs. One reason is that, while essential
self-adjointness can be characterized via $\al$-harmonic functions, it is 
not characterized via positive $\al$-harmonic functions. Thus, averaging
as in Lemma~\ref{l:averaging_energy} may produce a trivial function.

\section{Stability}\label{s:stability}
In this section, we study the stability of form uniqueness in a parallel
development to that for essential self-adjointness presented in \cite{IKMW25}.
More specifically, we show that if we can decompose a graph into two subgraphs
which are not too strongly connected, then form uniqueness of the entire graph
is equivalent to the form uniqueness of both graphs. See also \cite{Hua11, KL12}
for some results for stochastic completeness and \cite{HM24} for 
some recent work on recurrence. We then apply this stability result
to the case of graphs which consist of weakly spherically symmetric
components following the removal of a set of vertices, i.e., graphs with
weakly spherically symmetric ends. Finally, we give a series of examples
to illustrate our results.

\subsection{A stability criterion}
In this subsection, we decompose a graph into three pieces: two are induced subgraphs 
and one is a connecting graph.
If the connecting graph satisfies a boundedness property, 
we show that the graph satisfies form uniqueness if and only if
the two subgraphs satisfy form uniqueness as for essential self-adjointness in \cite{IKMW25}.

Let $(b,c)$ be a graph over $(X,m)$. We highlight that we do not
require local finiteness for the considerations of this subsection. Let $X_1 \subseteq X$ and
$X_2 = X \setminus X_1$. We then restrict $b$ to $X_i \times X_i$ and denote the restriction by $b_i$ for $i=1,2$. We further restrict $c$ and $m$ to $X_i$ and denote
the restrictions by $c_i$ and $m_i$ for $i=1,2$. This gives induced subgraphs
$(b_i,c_i)$ over $(X_i, m_i)$ for $i=1,2$.

Let $\pt X_i=\{ x \in X_i \mid \textup{ there exists } y \not \in X_i 
\textup{ with } y \sim x \}$ for $i=1,2$ and let
$X_3 = \pt X_1 \cup \pt X_2$. Let
$$b_\pt (x,y)= 
\begin{cases}
b(x, y) \quad & \textup{if } x \in \partial X_1, y \in X_2 \\
b(x, y) \quad & \textup{if } x \in \partial X_2, y \in X_1 \\
0 \quad &\text{otherwise}
\end{cases}$$
and restrict $m$ to $X_3$ to give $m_3$. This gives a graph $b_\pt$
over $(X_3,m_3)$.

We extend all functions by $0$ to the entire vertex set.
Note that the edge weight function decomposes as $b=b_1+b_2+b_\pt$
and the killing term as $c=c_1+c_2$.
On the energy form level, this implies
$$\Q_{b,c}=\Q_{b_1,c_1} + \Q_{b_2,c_2}+\Q_{b_\pt}$$
where all energy forms above can be thought
to act on the set $C(X)$.

In order to obtain our stability result, let
$\Deg_{\pt}\colon X_3 \to [0,\infty)$ be given by
$$\Deg_\pt(x)=\frac{1}{m(x)} \sum_{y \in X_3}b_\pt(x,y).$$
We assume that $\Deg_\pt$ is bounded. 
In particular, this implies that restrictions of $\Q_{b_\pt}$
are bounded, see Theorem~1.27 in \cite{KLW21}, so that
$$\QD_{b_{\pt}}=\QN_{b_{\pt}}$$
as the domain is the entire Hilbert space $\ell^2(X_3,m_3)$.

We will now prove the key lemma for our stability result. For this,
let $\Q_i=\Q_{b_i,c_i}$ and denote restrictions 
of these forms by $Q_i$ for $i=1,2$. Let
$$\D_i= \{ u \in C(X_i) \mid \Q_{i}(u)<\infty \}$$
so that $D(Q_i^{(N)})=\D_i \cap \ell^2(X_i,m_i)$
for $i=1,2$. 

As above, we extend all functions by $0$. In the other direction,
if $u \in C(X)$, then we let $u_i$ or  $u^{(i)}$ denote the restriction
of $u$ to $X_i$, i.e., $u|_{X_i}$ for $i=1,2$.

\begin{lemma}[Form decomposition]\label{l:stability}
Let $(b,c)$ be a graph over $(X,m)$. Let $X_1 \subseteq X$ and
$X_2=X \setminus X_1$ with 
$\Q_{b,c}=\Q_{1} + \Q_{2}+\Q_{b_\pt}$.
Let $\Deg_\pt$ be bounded. Then:
\begin{itemize}
\item[(a)] $u \in \ell^2(X, m)$ if and only if $u_i \in \ell^2(X_i,m_i)$ for $i=1,2$.
\item[(b)] $u \in D(\QN)$ if and only if $u_i \in D(\QN_i)$ for $i=1,2$.
\item[(c)] $\vp_n \to u$ 
in $\| \cdot \|_{\Q}$ for $\vp_n \in C_c(X)$ if and only
if $\vp_n^{(i)} \to u_i$ in $\|\cdot \|_{\Q_i}$ for $i=1,2$.
\end{itemize}
\end{lemma}
\begin{proof}
The proof of (a) is trivial, while (b) follows directly from
the decomposition $\Q_{b,c}=\Q_{1} + \Q_{2}+\Q_{b_\pt}$ and the fact
that $\Q_{b_\pt}$ is bounded since $\Deg_\pt$ is bounded and thus
all functions in $\ell^2(X_3,m_3)$ have finite energy.

For (c) note that if $\vp_n \to u$ in $\| \cdot \|_{\Q}$, then it is trivial
that $\vp_n^{(i)} \to u_i$ in $\|\cdot \|_{\Q_i}$ for $i=1,2$. Conversely,
if $\vp_n^{(i)} \to u_i$ in $\|\cdot \|_{\Q_i}$ for $i=1,2$, then clearly
$\vp_n^{(i)} \to u_i$ in $\ell^2(X_3, m_3)$ and, since $\QD_\pt$ is bounded on $\ell^2(X_3,m_3)$,
we get $\QD_\pt(\vp_n^{(i)} - u_i) \to 0$. Therefore, $\vp_n \to u$ in $\| \cdot \|_{\Q}$ from $\Q_{b,c}=\Q_{1} + \Q_{2}+\Q_{b_\pt}$.
\end{proof}

With the lemma above, it is now easy to give a stability result as follows.
\begin{theorem}[Stability of form uniqueness]\label{t:stability}
Let $(b,c)$ be a graph over $(X,m)$. Let $X_1 \subseteq X$ and
$X_2=X \setminus X_1$ with 
$\Q_{b,c}=\Q_{1} + \Q_{2}+\Q_{b_\pt}$.
Let $\Deg_\pt$ be bounded.
Then,
$$\QD=\QN \qquad \Longleftrightarrow \qquad \QD_i=\QN_i \quad \textup{ for } i=1,2.$$
\end{theorem}
\begin{proof}   
$\Longrightarrow$: Suppose $\QD=\QN$ and let $u \in D(\QN_1)$. Let $\ow{u}$ be the extension of $u$
to $X$ by 0. It is obvious that $\ow{u} \in \ell^2(X,m)$ and 
by $\Q_{b,c}=\Q_{1} + \Q_{2}+\Q_{b_\pt}$ we get $\ow{u} \in \D$ so that
$\ow{u} \in D(\QN).$ As $\QD=\QN$, there exists a sequence $\ow{\vp}_n \in C_c(X)$
satisfying $\ow{\vp}_n \to \ow{u}$ in $\| \cdot \|_\Q$. By Lemma~\ref{l:stability}~(c), $\vp_n = \ow{\vp}_n|_{X_1}$ satisfies $\vp_n \to u$ in 
$\|\cdot\|_{\Q_1}$. 
Thus, $\QD_1=\QN_1$. An analogous argument works for $u \in D(\QN_2)$. 

$\Longleftarrow$: Suppose that $\QD_i=\QN_i$ for $i=1,2$ and let
$u \in D(\QN)$. Let $u_i = u|_{X_i}$ so that
$u_i \in D(\QN_i)$ for $i=1,2$ by Lemma~\ref{l:stability}~(b).
As $\QD_i=\QN_i$ for $i=1,2$, there exist sequences $\vp_n^{(i)} \in C_c(X_i)$
such that $\vp_n^{(i)} \to u_i$ in $\|\cdot\|_{\Q_i}$ for $i=1,2$. 
Extending $\vp_n^{(i)}$ by 0 and letting
$\vp_n=\vp_n^{(1)}+\vp_n^{(2)}$, we now get $\vp_n \to u$ in $\|\cdot\|_\Q$ 
by Lemma~\ref{l:stability}~(c) so that $u \in D(\QD)$. 
\end{proof}

\begin{remark}
Similar stability results can be found for essential self-adjointness, stochastic
completeness, recurrence, and the Feller property for graphs, quantum
graphs, and manifolds
in, e.g., \cite{deTT11, Gri99, Hua11, HM24, IKMW25, KL12, KLW21, KMN22b, PS12, Woe00}.
\end{remark}

\subsection{Graphs with symmetric ends}
Given the stability criterion presented in the previous subsection
we now extend our characterizations to graphs which, following the removal
of a subset of vertices, decompose into disjoint unions of weakly spherically
symmetric graphs. We start by giving a definition. As in the previous
subsection, let $X_1 \subseteq X$ with $X_2=X \setminus X_1$.

\begin{definition}[Graphs with symmetric ends]
A graph $(b,c)$ over $(X,m)$ is a 
\emph{graph with (weakly spherically) symmetric ends (with respect to $X_1$ 
chosen for form uniqueness)} 
if $X_1 \subseteq X$ can be chosen so that
\begin{itemize}
\item[(1)] $\QD_1=\QN_1$.
\item[(2)] $\Deg_\pt$ is bounded.
\item[(3)] $(b_2,c_2)$ over $(X_2,m_2)$ is a disjoint union of weakly spherically
symmetric graphs.
\end{itemize}
\end{definition}

\begin{remark}[On the name ``graphs with symmetric ends'']
Note that there are various notions of ends for graphs,
e.g., \cite{DK03, Fre31, Hal64, KMN22b, Woe00}. The notion above
is in the spirit of Riemannian manifolds
and is tailored to our problem of studying graphs which, following the removal
of a subset of vertices, have connected components that are weakly spherically
symmetric and for which the form uniqueness problem can be localized
to these connected components. To be more precise, the defining 
phrase should reflect both
the dependency on $X_1$ and that $X_1$ is chosen in this way.
However, the phrase
``a graph with weakly spherically symmetric ends with respect to $X_1$ chosen for form uniqueness''
is too long, thus, we generally
shorten it to just a ``graph with symmetric ends'' or some
variant thereof.
\end{remark}

Note that condition (1) in the definition above is satisfied
if, for example, there exists an intrinsic pseudometric for which
the weighted vertex degree function is bounded on distance balls in $X_1$ or 
if the measure is uniformly bounded from below in $X_1$,
see \cite{HKMW13, KL12}. In particular, this is always the case
when $X_1$ is finite. Furthermore, 
if $X_1$ is finite and the graph is locally finite,
then (2) clearly holds as well.

Write $X_2=\bigsqcup_{i\in I}X_2^{(i)}$ with $(b_2^{(i)}, c_2^{(i)})$ 
being a weakly spherically symmetric graph over $(X_2^{(i)},m_2^{(i)})$
for $i \in I$.
Let
$q_2^{(i)}=c_2^{(i)}/m_2^{(i)}$ and denote the
energy
forms by $\Q_{2,i}$ as well as their restrictions by $\QD_{2,i}$ and $\QN_{2,i}$. 
Disjoint here means that there are no
edges between the weakly spherically symmetric graphs, i.e.,
$b(x^{(i)}, x^{(j)})=0$ if $i \neq j$ where $x^{(k)} \in X_2^{(k)}$
for $k \in I$.
We then have the following characterization.

\begin{theorem}[Form uniqueness for graphs with symmetric ends]\label{t:ends}
Let $(b,c)$ over $(X,m)$ be a graph with symmetric ends for form uniqueness.
Then the following statements are equivalent:
\begin{itemize}
\item[(i)] $\QD \neq \QN$.
\item[(ii)] $\QD_{2,i}\neq\QN_{2,i}$ \quad for some $i \in I$.
\item[(iii)] $(c+m)(X_2^{(i)}) < \infty$ \quad {and} \quad
$\displaystyle{\sum_{r=0}^\infty \frac{1}{\pt B_2^{(i)}(r)} < \infty}$
\quad for some $i \in I$.
\item[(iv)] $0 < \cp(\pt X_2^{(i)}) < \infty$ \quad for some $i \in I$.
\end{itemize} 
\end{theorem}

\begin{proof}
By Theorem~\ref{t:stability}, $\QD \neq \QN$ if and only if $\QD_2 \neq \QN_2$.
Now, as the weakly spherically symmetric graphs making up $(b_2,c_2)$ over 
$(X_2,m_2)$ are disjoint, it follows that $\QD_2 \neq \QN_2$ if and only if
$\QD_{2,i}\neq\QN_{2,i}$ for some $i \in I$ as can be seen by using $\al$-harmonic 
functions (Lemma~\ref{l:al-harmonic}). This gives the equivalence 
between (i) and (ii).
The equivalence between (ii), (iii), and (iv) follows from 
Theorems~\ref{t:form_uniqueness_weights}~and~\ref{t:capacity}.
\end{proof}

\begin{remark}
We contrast our result for graphs with symmetric ends to recent
work on quantum graphs in \cite{KMN22b}. There, the authors
consider topological ends of quantum graphs which, in the locally
finite case, can be thought
of either as equivalence classes of rays or as connected components
when removing finite sets in an exhaustion of the graph. For such
ends they show that form uniqueness is equivalent to the fact that
every end has infinite volume, see Corollary~3.13 in \cite{KMN22b}.
%In contrast, in our result above, we see that there are three distinct
%ways that form uniqueness can happen on an end with respect to $X_1$:
%via the measure $m$, the killing term $c$ or the edge
%weight $b$.
\end{remark}

\subsection{Examples}
In this final subsection we give a series of examples to illustrate
our results from the preceding sections.

We first give the simplest example to illustrate Theorem~\ref{t:ends}
on graphs with symmetric ends 
by gluing two birth--death chains together.
In particular, this allows us to create examples of graphs that 
do not satisfy form uniqueness but have infinite measure.
Note that most existing examples of non-form unique graphs are constructed
by taking a transient graph, transience being a property
that is independent of the vertex measure, and changing the vertex
measure to be finite. However, see Example~5.6 in \cite{HKMW13}
which contains a graph as in the example below.

\begin{example}[Bilateral birth--death chains]\label{ex:bilateral}
Let $X=\Z$ with $b(x,y)>0$ if and only if $|x-y|=1$ and $c=0$. 
Such graphs are sometimes
called bilateral birth--death chains. 

Let $X_1=\{0\}$ and note that, by Theorem~\ref{t:ends}, $\QD \neq \QN$ if and only if either
$$m(\N)<\infty \qquad \textup{and} \qquad \sum_{r=0}^\infty \frac{1}{b(r,r+1)}<\infty$$
or
$$m(-\N)<\infty \qquad \textup{and} \qquad \sum_{r=0}^\infty \frac{1}{b(-r,-r-1)}<\infty.$$

Thus, for example, choosing $m(\N)<\infty$, $\sum_r {1}/{b(r,r+1)}<\infty$
and $m(-\N)=\infty$ gives a graph with $\QD\neq \QN$ and $m(X)=\infty$.
\end{example}

Next, we give an example to illustrate our capacity
result, Theorem~\ref{t:capacity}. As mentioned previously, 
if the Cauchy boundary has infinite capacity, there exist examples
of graphs which both satisfy and do not satisfy form uniqueness,
see Examples~5.2~and~5.4 in \cite{HKMW13}. However, the example
of non-form uniqueness given in \cite{HKMW13} is a bilateral
birth--death chain for which the Cauchy boundary consists
of two points. Thus, one may ask if form uniqueness
always follows if the Cauchy boundary consists of one point
and has infinite capacity. We now show 
by using our stability theory that this is not the case.

\begin{figure}[htbp]
\centering
\begin{tikzpicture}[main_node/.style={circle,draw,minimum size=0.5em, fill, inner sep=2pt]}]
\node[main_node, label=above:0] (0) at (-5, 5) {};
\node[main_node, label=above:1] (1) at (-3.5, 5) {};
\node[main_node, label=above:2] (2) at (-2, 5) {};
\node[main_node, label=above:3] (3) at (-0.5, 5) {};
\node[main_node, label=below:$x_0$] (4) at (-5, 3.5) {};
\node[main_node, label=below:$x_1$] (5) at (-3.5, 3.5) {};
\node[main_node, label=below:$x_2$] (6) at (-2, 3.5) {};
\node[main_node, label=below:$x_3$] (7) at (-0.5, 3.5) {};

\node[right=0.8cm of 3] (8) {};

\draw node[fill,circle,minimum size=0.1em, inner sep=0.5pt] at (0.8, 4.25){};
\draw node[fill,circle,minimum size=0.1em, inner sep=0.5pt] at (1, 4.25){};
\draw node[fill,circle,minimum size=0.1em, inner sep=0.5pt] at (1.2, 4.25){};

 \path[draw, thick]
(0) edge node {} (4) 
(0) edge node {} (1) 
(1) edge node {} (5) 
(1) edge node {} (2) 
(2) edge node {} (6) 
(3) edge node {} (7) 
(2) edge node {} (3)
(3) edge node {} (8) {};

% \draw (3) node {} -- (8) [dashed, thick] node {};

\end{tikzpicture}

\caption{The graph in Examples~\ref{e:Cauchy_boundary}~and~\ref{ex:instability}~(1).}
\label{fig:pendant}
\end{figure}
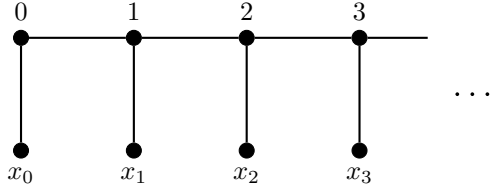

\begin{example}[One point Cauchy boundary with infinite capacity
but no form uniqueness]\label{e:Cauchy_boundary}
Let $X_1 = \N_0 =\{0,1,2, \ldots\}$ and $X_2 =\{x_0, x_1, x_2, \ldots\}$
with $X= X_1 \cup X_2$. Let $b$ be defined symmetrically so that 
\begin{itemize}
\item{\makebox[2.5cm]{$b(k,l)>0$ \hfill} if and only if  $|k-l|=1$
for $k, l \in \N_0$,}
\item{\makebox[2.5cm]{$b(k,x_k)>0$ \hfill} for $k \in \N_0,$}
\item{\makebox[2.5cm]{$b=0$ \hfill} otherwise,}
\end{itemize}
and let $c=0$. Thus, the graph may be visualized
as a birth--death chain where each vertex has an additional 
pendant vertex attached to it, see Figure~\ref{fig:pendant}.

Let $b$ and $m$ be defined on $X_1$ so that $\QD_1 \neq \QN_1$.
That is, by Theorem~\ref{t:form_uniqueness_weights}, let 
$$m(X_1)<\infty \qquad \textup{and} \qquad \sum_{k=0}^\infty \frac{1}{b(k,k+1)}<\infty.$$ 
Furthermore, let $m(x_k)=1$ and $b(k,x_k)=m(k)$ for all $k \in \N_0$.
Therefore, $\Deg_\pt(k)=b(k,x_k)/m(k)=1$ and $\Deg_\pt(x_k)= b(x_k,k) = m(k) \to 0$
as $k \to \infty$ as $m(X_1)<\infty$.
By Theorem~\ref{t:stability} it now follows that $\QD\neq \QN$.

Finally, we define a strongly intrinsic path metric and show
that the Cauchy boundary consists of a single point which has infinite
capacity. For this, we recall that a standard strongly intrinsic
path metric is given by choosing
$$\si(x,y) = \min\left\{\Deg(x)^{-1/2},\Deg(y)^{-1/2}\right\}$$
for $x \sim y$ where $\Deg$ is the weighted
vertex degree, see Example~\ref{ex:intrinsic}~(2) above.
We note that $\pt_\si X \neq \emptyset$ as if $\pt_\si X = \emptyset$, then $\QD=\QN$
since $\pt_\si X = \emptyset$ would imply the essential self-adjointness of $\De_c$, see \cite{HKMW13}.

By the assumptions that
$m(X_1)<\infty$ and $\sum_k{1}/{b(k,k+1)}<\infty,$
it is clear that $\Deg(k) \to \infty$ and thus 
$$\si(k,k+1) \to 0 \qquad \textup{and} \qquad \si(k,x_k) \to 0$$ 
as $k \to \infty$ so that
$\pt_\si X$ consists of a single accumulation point. Finally,
$\cp(\pt_\si X) =\infty$ since any open neighborhood
of the boundary will contain infinitely many of the vertices $x_k$
and $m(x_k)=1$ for every $k \in \N_0$.
\end{example}

Finally, we give a series of instability examples
which show that the boundedness assumption on $\Deg_\pt$ in the 
stability results in Theorems~\ref{t:stability}~and~\ref{t:ends}
is necessary. 

\begin{example}[Instability]\label{ex:instability}
The first two examples below fall under the general framework
of Theorem~\ref{t:stability} and involve starting
with a non-form unique birth--death chain to which
we attach either pendant vertices or an infinite star to make the entire 
graph satisfy form uniqueness.
In the third example, we illustrate Theorem~\ref{t:ends}
on graph with symmetric ends
by starting with a form unique core
birth--death chain and attaching two birth--death chains,
one of which is not form unique, in such a way that the entire
graph is form unique.
For all of the examples, the proof technique is the same: we
analyze strictly positive $\al$-harmonic functions for which every vertex
must have a neighboring vertex where the $\al$-harmonic function is
strictly bigger and showing that they cannot have finite energy
under some conditions on the graph.\\

(1) \textbf{Birth--death chain and pendant vertices:} 
The first example is a variant on Example~\ref{e:Cauchy_boundary}.
Let $X_1 = \N_0 =\{0,1,2, \ldots\}$ and $X_2 =\{x_0, x_1, x_2, \ldots\}$
with $X= X_1 \cup X_2$. Let $b$ be defined symmetrically so that
\begin{itemize}
\item{\makebox[2.5cm]{$b(k,l)>0$ \hfill} if and only if  $|k-l|=1$
for $k, l \in \N_0,$}
\item{\makebox[2.5cm]{$b(k,x_k)>0$ \hfill} for $k \in \N_0,$}
\item{\makebox[2.5cm]{$b=0$ \hfill} otherwise,}
\end{itemize}
and let $c=0$. Thus, the graph may be visualized
as a birth--death chain where each vertex has an additional 
pendant vertex attached to it, see Figure~\ref{fig:pendant}.

Let $b$ and $m$ be defined on $X_1$ so that $\QD_1 \neq \QN_1$.
That is, by Theorem~\ref{t:form_uniqueness_weights}, let 
$$m(X_1)<\infty \qquad \textup{and} \qquad \sum_{k=0}^\infty \frac{1}{b(k,k+1)}<\infty.$$ 

We now show that we can
extend $b$ and $m$ to $X_2$ so that $\QD=\QN$. 
For this, we let $u>0$ satisfy $(\De +1)u=0$
and show that we can choose $b$ and $m$ so that $u \not \in \D$.

From $(\De +1)u(x_k)=0$ we get
$$u(x_k)=\left(\frac{b(x_k,k)}{b(x_k,k)+ m(x_k)}\right)u(k) < u(k)$$
and thus
\begin{equation}\label{ex:energy}
b(x_k,k)(u(x_k)-u(k))^2 = \frac{ m^2(x_k)}{b(x_k,k)}u^2(x_k) = 
\frac{ b(x_k,k)m^2(x_k)}{\left(b(x_k,k)+ m(x_k)\right)^2}u^2(k).
\end{equation}
In particular, as $u(x_k)< u(k)$ for all $k \in \N_0$,
it follows, by using $(\De +1)u(k)=0$ and induction, that
$u(k) < u(k+1)$ for all $k \in \N_0$.
See Figure~\ref{fig:pendant_directed} for an illustration of
this argument.

Therefore, if
$$\sum_{k=0}^\infty \frac{b(x_k,k)m^2(x_k)}{\left(b(x_k,k)+ m(x_k)\right)^2}=\infty,$$
then $u \not \in \D$ by \eqref{ex:energy} and
thus $\QD=\QN$ by Lemma~\ref{l:al-harmonic}.
For a concrete example, let $b(x_k,k)=1=m(x_k)$ in \eqref{ex:energy}, 
to get 
$$(u(x_k)-u(k))^2 = \frac{u^2(k)}{4} >\frac{u^2(0)}{4} $$
for all $k \in \N$.

Note that, in this case, $X_2 = \{x_0, x_1, x_2, \ldots\}$ with $b_2=0$,
thus giving a disconnected graph. Furthermore, note that 
$\Deg_\pt(k)={b(x_k,k)}/{m(k)}$ and if this quantity
were bounded, then, from the assumption that $m(X_1)<\infty$, we would
obtain $\sum_k b(x_k,k)<\infty$. This implies
$$\sum_{k=0}^\infty \frac{b(x_k,k)m^2(x_k)}{\left(b(x_k,k)+m(x_k)\right)^2}
< \sum_{k=0}^\infty b(x_k,k)<\infty$$
which contradicts the assumption on $b$.
Thus, $\Deg_\pt(k)$ must be unbounded in this case
as follows by Theorem~\ref{t:stability}.

\begin{figure}[htbp]
%\centering

\begin{tikzpicture}[main_node/.style={circle,draw,minimum size=0.5em, fill, inner sep=2pt]}]
\tikzset{edge/.style = {->,> = latex'}}

\node[main_node, label=above:0] (0) at (-5, 5) {};
\node[main_node, label=above:1] (1) at (-3.5, 5) {};
\node[main_node, label=above:2] (2) at (-2, 5) {};
\node[main_node, label=above:3] (3) at (-0.5, 5) {};
\node[main_node, label=below:$x_0$] (4) at (-5, 3.5) {};
\node[main_node, label=below:$x_1$] (5) at (-3.5, 3.5) {};
\node[main_node, label=below:$x_2$] (6) at (-2, 3.5) {};
\node[main_node, label=below:$x_3$] (7) at (-0.5, 3.5) {};

\draw node[fill,circle,minimum size=0.1em, inner sep=0.5pt] at (0.8, 4.25){};
\draw node[fill,circle,minimum size=0.1em, inner sep=0.5pt] at (1, 4.25){};
\draw node[fill,circle,minimum size=0.1em, inner sep=0.5pt] at (1.2, 4.25){};

\node[right=1cm of 3] (8) {};

\draw[thick, mid arrow] (4) to (0);
\draw[thick, mid arrow] (0) to (1);
\draw[thick, mid arrow] (5) to (1);
\draw[thick, mid arrow] (1) to (2);
\draw[thick, mid arrow] (6) to (2);
\draw[thick, mid arrow] (2) to (3);
\draw[thick, mid arrow] (7) to (3);
\draw[thick, mid arrow] (3) to (8);

\end{tikzpicture}
\caption{The graph in Example~\ref{ex:instability}~(1).
The arrows indicate the directions of increase for any positive $\al$-harmonic
function.}
\label{fig:pendant_directed}
\end{figure}
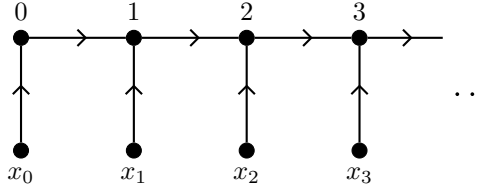

\bigskip
(2) \textbf{Birth--death chain and infinite star:} 
We now give an example, based on
the previous one, where $X_2$ is also a connected subgraph. For this, 
let $X_1 = \N_0 =\{0,1,2, \ldots\}$ and $X_2 =\{o, x_0, x_1, x_2, \ldots\}$
with $X= X_1 \cup X_2$. 
Let $b$ be defined symmetrically so that
\begin{itemize}
\item{\makebox[2.5cm]{$b(k,l)>0$ \hfill} if and only if  $|k-l|=1$
for $k, l \in \N_0$},
\item{\makebox[2.5cm]{$b(k,x_k)>0$ \hfill} for $k \in \N_0,$}
\item{\makebox[2.5cm]{$b(o, x_k)>0$ \hfill} for $k \in \N_0$}
with $\sum_k b(o,x_k)<\infty$,
\item{\makebox[2.5cm]{$b=0$ \hfill} otherwise,}
\end{itemize}
and let $c=0$. Thus, the graph may be visualized
as a birth--death chain where each vertex has an additional 
vertex attached to it, and all additional vertices are connected to one
vertex, see Figure~\ref{fig:star}. 
In particular, $b_2$ over $(X_2, m_2)$ is an infinite star graph and is thus connected.

As above, let 
$$m(X_1)<\infty \qquad \textup{and} \qquad \sum_{k=0}^\infty \frac{1}{b(k,k+1)}<\infty$$ 
so that $\QD_1 \neq \QN_1$. We 
will extend these choices to $b$ and $m$ so that $\QD=\QN$.
For this, it suffices that we take $b$ and $m$
so that 
$$\inf_{k \in \N_0} m(x_k)>0 \qquad \textup{and} \qquad \sum_{k=0}^\infty b(k, x_k) =\infty$$
as we will now show.

We let $u>0$ satisfy $(\De +\al)u=0$ for $\al>0$ and show
that $u \not \in \D \cap \ell^2(X,m)$.
If $u \in \ell^2(X,m)$, then, from the assumption
$\inf_{k \in \N_0} m(x_k)>0$, we obtain $u(x_k) \to 0$ as $k \to \infty$.
Therefore, $u(x_k) \leq u(o)$ for all $k$ large enough.
We will show that this allows us to push all the growth of the function
into the birth--death chain part of the graph which then gives a contradiction
to the finite energy of $u$ on the vertical edges.

\eat{
As an illustrative simple first case,
we now let $v(x_0) \geq v(x_k)$ for all $k \in \N$.
In this case, from $(\De+\al)v(x_0)=0$, we obtain that $v(x_0)<v(0)$
and then from $(\De+\al)v(0)=0$, we likewise get $v(0)<v(1)$.
Therefore, $v(1)>v(0)>v(x_0) \geq v(x_1)$ so that $v(1)>v(x_1)$.
Then, $(\De+\al)v(1)=0$ implies $v(1)<v(2)$ and by induction
we may prove that $v(k)<v(k+1)$ for all $k \in \N_0$.

Thus, since $v(k)$ is increasing while $v(x_0) \geq v(x_k)$, we now
get
$$b(k,x_k)(v(k)-v(x_k))^2 \geq b(k,x_k)(v(0)-v(x_0))^2$$
for all $k \in \N_0$ and from $\sum_k b(k,x_k)=\infty$, we get
a contradiction to $v \in \D$. Therefore, $v=0$ and thus $\QD=\QN$
in this case.}

Let $k_0\in \N_0$ be 
such that $u(x_k) \leq u(o)$ for all $k > k_0$ while
$u(x_{k_0})> u(o)$. 
Note that there exists at least one $k$ such that $u(x_k)> u(o)$ from
$(\De+\al)u(o)=0$.
Then, $u(x_{k_0}) < u(k_0)$ as follows from $(\De+\al)u(x_{k_0})=0$. 
Now, if $u(k_0) \geq u(k_0+1)$, then,
from $(\De+\al)u(k_0)=0$, we would get $u(k_0) < u(k_0 -1)$.
Thus, starting at the vertex $k_0$, the function $u$ could now only increase along a path
consisting of vertices in $\{0,1, \ldots, k_0, o, x_0, x_1, \ldots, x_{k_0}\}.$
As this is a finite set, we eventually reach a contradiction since
we run out of vertices along which the function may increase. Therefore,
$u(k_0) < u(k_0+1)$.

Now, 
$$u(x_{k_0+1}) \leq u(o) < u(x_{k_0})< u(k_0)<u(k_0+1)$$ 
and thus  $u(k_0+1) < u(k_0+2)$ from $(\De+\al)u(k_0+1)=0$.
Iterating the argument and using induction we get
$u(k)< u(k+1)$ for all $k \geq k_0$, see Figure~\ref{fig:star}
for a visualization of this argument.

Since $u(x_k) \leq u(o)$ and $u(k)\geq u(k_0)$ for all $k \geq k_0$, it now follows that
$$b(k,x_k)(u(k)-u(x_k))^2 \geq b(k,x_k)(u(k_0)-u(o))^2$$
for all $k \geq k_0$. Thus, $\sum_k b(k,x_k)=\infty$ gives $u \not \in \D$. Therefore, $\QD=\QN$.

Note that $\Deg_\pt (k)=b(k,x_k)/m(k)$
and if this were bounded, then, from the assumption $m(X_1)<\infty$, we would obtain
$\sum_k b(k,x_k)<\infty$ which contradicts the assumption on $b$. Thus,
$\Deg_\pt (k)$ must be unbounded as follows by Theorem~\ref{t:stability}.  \\

\begin{figure}[htbp]
\begin{tikzpicture}[main_node/.style={circle,draw,minimum size=0.5em, fill, inner sep=2pt]}]

\node[main_node, label=above:0] (0) at (-5, 5) {};
\node[main_node, label=above:1] (1) at (-3.5, 5) {};
\node[main_node, label=above:2] (2) at (-2, 5) {};
\node[main_node, label=right:$x_0$] (x_0) at (-5, 3.5) {};
\node[main_node, label=right:$x_1$] (x_1) at (-3.5, 3.5) {};
\node[main_node, label=right:$x_2$] (x_2) at (-2, 3.5) {};
\node[right=0.8cm of 2] (3) {};

\node[main_node, label=below:$o$] (o) at (1, 1) {};

\node[main_node, label=above:$k_0-1$] (k_0-1) at (1, 5) {};
\node[main_node, label=right:$x_{k_0-1}$] (x_{k_0-1}) at (1, 3.5) {};
\node[main_node, label=above:$k_0$] (k_0) at (2.5, 5) {};
\node[main_node, label=right:$x_{k_0}$] (x_{k_0}) at (2.5, 3.5) {};
\node[main_node, label=above:$k_0+1$] (k_0+1) at (4, 5) {};
\node[main_node, label=right:$x_{k_0+1}$] (x_{k_0+1}) at (4, 3.5) {};
\node[main_node, label=above:$k_0+2$] (k_0+2) at (5.5, 5) {};
\node[main_node, label=right:$x_{k_0+2}$] (x_{k_0+2}) at (5.5, 3.5) {};
\node[right=0.8cm of k_0+2] (k_0+3) {};
\node[left=0.8cm of k_0-1] (k_0-2) {};

\draw node[fill,circle,minimum size=0.1em, inner sep=0.5pt] at (-0.5, 4.25){};
\draw node[fill,circle,minimum size=0.1em, inner sep=0.5pt] at (-0.3, 4.25){};
\draw node[fill,circle,minimum size=0.1em, inner sep=0.5pt] at (-0.1, 4.25){};

\draw node[fill,circle,minimum size=0.1em, inner sep=0.5pt] at (6.8, 4.25){};
\draw node[fill,circle,minimum size=0.1em, inner sep=0.5pt] at (7, 4.25){};
\draw node[fill,circle,minimum size=0.1em, inner sep=0.5pt] at (7.2, 4.25){};

 \path[draw, thick]
(0) edge node {} (1) 
(1) edge node {} (2)  
(0) edge node {} (x_0) 
(1) edge node {} (x_1) 
(2) edge node {} (x_2)
(k_0-1) edge node {} (k_0)
(k_0-1) edge node {} (x_{k_0-1})
(o) edge node {} (x_0) 
(o) edge node {} (x_1)
(o) edge node {} (x_2) 
(o) edge node {} (x_{k_0-1})
(k_0-2) edge node {} (k_0-1);

\draw[thick] (2) to (3);
\draw[thick, mid arrow] (x_{k_0}) to (k_0);
\draw[thick, mid arrow] (k_0) to (k_0+1);
\draw[thick, mid arrow] (x_{k_0+1}) to (k_0+1);
\draw[thick, mid arrow] (k_0+1) to (k_0+2);
\draw[thick, mid arrow] (x_{k_0+2}) to (k_0+2);
\draw[thick, mid arrow] (k_0+2) to (k_0+3);
\draw[thick, mid arrow] (o) to (x_{k_0});
\draw[thick, mid arrow] (x_{k_0+1}) to (o);
\draw[thick, mid arrow] (x_{k_0+2}) to (o);

\end{tikzpicture}

\caption{The graph in Example~\ref{ex:instability}~(2).
The arrows indicate the directions of increase for $u$.}
\label{fig:star}
\end{figure}

\eat{
(3) \textbf{Infinite ladder:}
We next give an instability example constructed
out of two birth--death chains glued vertically, i.e., an infinite
ladder graph.
We let $X_1 = \N_0 =\{0,1,2, \ldots\}$ and $X_2 =\{x_0, x_1, x_2, \ldots\}$
with $X= X_1 \cup X_2$. 
Let $b$ be given by 
\begin{itemize}
\item $b(k,l)>0$ and $b(x_k, x_l)>0$ \qquad if and only if \qquad $|k-l|=1$
for $k, l \in \N_0$
\item $b(k,x_k)=b(x_k,k)>0$ \qquad for $k \in \N_0$
\item $b$ is zero otherwise
\end{itemize}
and let $c=0$. Thus, the graph may be visualized
as two birth--death chains connected vertically to form
an infinite ladder.

We let 
$$m(X_1)<\infty \qquad \textup{and} \qquad \sum_{r=0}^\infty \frac{1}{b(r,r+1)}<\infty$$ 
so that $\QD_1 \neq \QN_1$
and will extend these choices to $b$ and $m$ so that $\QD=\QN$.
For this, as we will show, it suffices that we take $b$ and $m$
so that 
$$m(X_2)=\infty \qquad \textup{and} \qquad \sum_{k=0}^\infty b(k, x_k) =\infty.$$

We let $u>0$ satisfy $(\De +\al)u=0$ for $\al>0$ and show
that $u \not \in \D \cap \ell^2(X,m)$.
If $u \in \ell^2(X,m)$ then, from the assumption
$m(X_2)=\sum_k m(x_k)=\infty$, we obtain $u(x_k) \to 0$ as $k \to \infty$.
Thus, let $k_0 \in \N_0$ be the largest index such that 
$$u(x_{k_0})= \max_{k \in \N_0}u(x_k).$$

Now, from $(\De+\al)u(x_{k_0})=0$ and the fact that $k_0$ is the largest
index where $u(x_k)$ achieves a maximum, we get
$u(x_{k_0})<u(k_0)$.
Therefore, $u(x_k) < u(k_0)$ for all $k \in \N_0$.

Now, if $u(k_0) \geq u(k_0 +1)$, then, from 
$(\De+\al)u(x_{k_0})=0$, we would get $u(k_0)< u(x_{k_0}-1)$
and thus, starting at $k_0$,
the function $u$ could now only increase along some path
consisting of vertices in the finite set 
$\{0,1, \ldots, k_0, x_0, x_1, \ldots, x_{k_0}\}$
which leads to a contradiction. Therefore,
$u(k_0) < u(k_0+1)$.

Now, $u(x_{k_0+1}) <u(k_0) <u(k_0+1) $ and thus
$(\De+\al)u(k_0+1)=0$ now gives
$u(k_0+1) < u(k_0+2)$.
Iterating this argument and using induction now implies
$u(k) < u(k+1)$
for all $k \geq k_0$.
Hence,
$$b(k,x_k)(u(k)-u(x_k))^2 \geq b(k,x_k)(u(k_0)-u(x_{k_0}))^2$$
for all $k \geq k_0$ and thus $\sum_k b(k,x_k)=\infty$ implies 
$u \not \in \D$. Therefore, $\QD=\QN$.

As in the first two examples, we note that $\Deg_\pt (k)=b(k,x_k)/m(k)$
and if this is bounded, then from $m(X_1)<\infty$ we would obtain
$\sum_k b(k,x_k)<\infty$ which contradicts our assumption. Thus,
$\Deg_\pt (k)$ must be unbounded. \\ }

(3) \textbf{Double infinite ladder:}
In this last example we illustrate Theorem~\ref{t:ends}
by connecting two birth--death chains, one of which is not form unique,
via a middle birth--death
chain which satisfies form uniqueness in such a way
that the entire graph is form unique.

Let $X_1=\{y_k\}_{k \in \N_0}$
and $X_2=\{x_k, z_k\}_{k \in \N_0}$. 
Let $b$ be defined symmetrically so that
\begin{itemize}
\item{\makebox[7cm]{$b(x_k,x_l)>0, \ b(y_k,y_l)>0, \ b(z_k,z_l)>0$ \hfill} iff  $|k-l|=1$
for $k, l \in \N_0,$}
\item{\makebox[7cm]{$b(x_k, y_k)>0, b(y_k, z_k) >0$ \hfill} for $k \in \N_0,$}
\item{\makebox[7cm]{$b=0$ \hfill} otherwise,}
\end{itemize}
and let $c=0$. Thus, the graph may be visualized
as an infinite double ladder, see Figure~\ref{fig:ladder}.

In order to fit the framework of Theorem~\ref{t:ends}, we choose
$b$ and $m$ so that $b_1$ over $(X_1,m_1)$ satisfies form uniqueness
by letting $\inf_{k \in \N_0} m(y_k)>0.$

We now choose $b$ and $m$ so that one of the birth--death chains
in $X_2$ is not form unique but so that the entire graph does
satisfy form uniqueness.
More specifically, let
$$\sum_{k=0}^\infty m(x_k)<\infty \qquad \textup{and} \qquad
\sum_{k=0}^\infty \frac{1}{b(x_k,x_{k+1})}<\infty$$
so that $b$ over $(\{x_k\}, m)$ is not form unique, let
$\inf_{k\in \N_0} m(z_k)>0$
so that $b$ over $(\{z_k\},m)$ is form unique, and
$\sum_{k=0}^\infty b(x_k,y_k)= \infty$
so that $b$ over $(X,m)$ will be form unique as we will now show.

We let $u>0$ satisfy $(\De +\al)u=0$ for $\al>0$ and show
that $u \not \in \D \cap \ell^2(X,m)$.
If $u \in \ell^2(X,m)$, then, from the assumptions
$\inf_{k\in \N_0} m(z_k)>0$ and $\inf_{k \in \N_0} m(y_k)>0,$ 
we obtain $u(z_k) \to 0$ and $u(y_k) \to 0$ as $k \to \infty$.
Let $k_0, k_1 \in \N_0$ be the largest indices such that 
$$u(z_{k_0})= \max_{k \in \N_0}u(z_k) \qquad \textup{and} \qquad
u(y_{k_1})= \max_{k \in \N_0}u(y_k).$$

Now from, $(\De+\al)u(z_{k_0})=0$ and the fact that $k_0$ is the largest
index where $u(z_k)$ achieves a maximum, we get
$u(z_{k_0})<u(y_{k_0})\leq u(y_{k_1})$. From this,
$$u(z_k) < u(y_{k_1})$$ 
for all $k \in \N_0$.
We repeat this argument at $y_{k_1}$, using $(\De+\al)u(y_{k_1})=0$
to obtain $u(y_{k_1})< u(x_{k_1})$ so that
$$u(y_k) < u(x_{k_1})$$ 
for all $ k \in \N_0$.

Now, if $u(x_{k_1}) \geq u(x_{k_1+1})$, then 
from $(\De+\al)u(x_{k_1})=0$ and $u(x_{k_1})>u(y_{k_1})$, we get
$u(x_{k_1}) < u(x_{k_1-1})$. Then, starting at $x_{k_1}$,
the function $u$ could now only increase along a path
consisting of vertices in the finite set 
$$\{x_0, x_1, \ldots x_{k_1}\}$$
which leads to a contradiction. Therefore,
$u(x_{k_1}) < u(x_{k_1+1})$.

Using $u(y_k) \leq u(y_{k_1}) < u(x_{k_1}) < u(x_{k_1+1})$
for all $k \in \N_0$ and induction we now get $u(x_{k}) < u(x_{k+1})$
for all $k \geq k_1$.
See Figure~\ref{fig:ladder} for a visualization of the argument.

Therefore, 
$$b(x_k, y_k)(u(x_k)-u(y_k))^2 \geq b(x_{k}, y_{k})(u(x_{k_1})-u(y_{k_1}))^2 $$
for all $k \geq k_1$ and $\sum_k b(x_k,y_k)=\infty$
now gives $u \not \in \D$. This shows $\QD=\QN$.

Note that $\Deg_\pt (x_k) = b(x_k, y_k)/m(x_k)$
and, if this were bounded, then the assumption $\sum_k m(x_k)<\infty$ implies
$\sum_k b(x_k,y_k)<\infty$ which contradicts the assumption on $b$. Thus,
$\Deg_\pt (x_k)$ must be unbounded, as follows from Theorem~\ref{t:ends}.

\eat{Finally, we note that there is nothing particular about the fact that we use
a double ladder in this example. We could, just as well, use a single ladder
or a ladder consisting of any finite number of birth--death chains. 
We choose two in order to illustrate the statement of Theorem~\ref{t:ends}.}

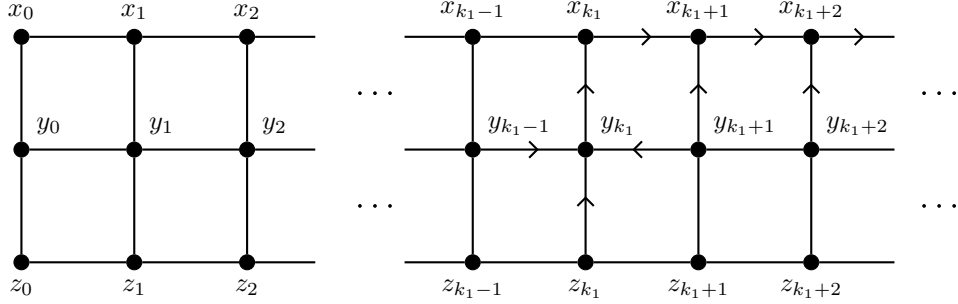
\begin{figure}[htbp] 
\begin{tikzpicture}[main_node/.style={circle,draw,minimum size=0.5em, fill, inner sep=2pt]}]

\node[main_node, label=above:$x_0$] (x_0) at (-5, 5) {};
\node[main_node, label=above:$x_1$] (x_1) at (-3.5, 5) {};
\node[main_node, label=above:$x_2$] (x_2) at (-2, 5) {};
\node[main_node, label=above right:$y_0$] (y_0) at (-5, 3.5) {};
\node[main_node, label=above right:$y_1$] (y_1) at (-3.5, 3.5) {};
\node[main_node, label=above right:$y_2$] (y_2) at (-2, 3.5) {};
\node[main_node, label=below:$z_0$] (z_0) at (-5, 2) {};
\node[main_node, label=below:$z_1$] (z_1) at (-3.5, 2) {};
\node[main_node, label=below:$z_2$] (z_2) at (-2, 2) {};
\node[right=0.8cm of x_2] (x_3) {};
\node[right=0.8cm of y_2] (y_3) {};
\node[right=0.8cm of z_2] (z_3) {};

\draw node[fill,circle,minimum size=0.1em, inner sep=0.5pt] at (-0.5, 4.25){};
\draw node[fill,circle,minimum size=0.1em, inner sep=0.5pt] at (-0.3, 4.25){};
\draw node[fill,circle,minimum size=0.1em, inner sep=0.5pt] at (-0.1, 4.25){};
\draw node[fill,circle,minimum size=0.1em, inner sep=0.5pt] at (-0.5, 2.75){};
\draw node[fill,circle,minimum size=0.1em, inner sep=0.5pt] at (-0.3, 2.75){};
\draw node[fill,circle,minimum size=0.1em, inner sep=0.5pt] at (-0.1, 2.75){};

\draw node[fill,circle,minimum size=0.1em, inner sep=0.5pt] at (7, 4.25){};
\draw node[fill,circle,minimum size=0.1em, inner sep=0.5pt] at (7.2, 4.25){};
\draw node[fill,circle,minimum size=0.1em, inner sep=0.5pt] at (7.4, 4.25){};
\draw node[fill,circle,minimum size=0.1em, inner sep=0.5pt] at (7, 2.75){};
\draw node[fill,circle,minimum size=0.1em, inner sep=0.5pt] at (7.2, 2.75){};
\draw node[fill,circle,minimum size=0.1em, inner sep=0.5pt] at (7.4, 2.75){};

\node[main_node, label=above:$x_{k_1-1}$] (x_{k-1}) at (1, 5) {};
\node[main_node, label=above:$x_{k_1}$] (x_k) at (2.5, 5) {};
\node[main_node, label=above:$x_{k_1+1}$] (x_{k+1}) at (4, 5) {};
\node[main_node, label=above:$x_{k_1+2}$] (x_{k+2}) at (5.5, 5) {};
\node[main_node, label=above right:$y_{k_1-1}$] (y_{k-1}) at (1, 3.5) {};
\node[main_node, label=above right:$y_{k_1}$] (y_k) at (2.5, 3.5) {};
\node[main_node, label=above right:$y_{k_1+1}$] (y_{k+1}) at (4, 3.5) {};
\node[main_node, label=above right:$y_{k_1+2}$] (y_{k+2}) at (5.5, 3.5) {};
\node[main_node, label=below:$z_{k_1-1}$] (z_{k-1}) at (1, 2) {};
\node[main_node, label=below:$z_{k_1}$] (z_k) at (2.5, 2) {};
\node[main_node, label=below:$z_{k_1+1}$] (z_{k+1}) at (4, 2) {};
\node[main_node, label=below:$z_{k_1+2}$] (z_{k+2}) at (5.5, 2) {};

\node[left=0.8cm of x_{k-1}] (x_{k-2}) {};
\node[left=0.8cm of y_{k-1}] (y_{k-2}) {};
\node[left=0.8cm of z_{k-1}] (z_{k-2}) {};
\node[right=1cm of x_{k+2}] (x_{k+3}) {};
\node[right=1cm of y_{k+2}] (y_{k+3}) {};
\node[right=1cm of z_{k+2}] (z_{k+3}) {};

 \path[draw, thick]
(x_0) edge node {} (x_1) 
(x_1) edge node {} (x_2)
(x_2) edge node {} (x_3)
(y_0) edge node {} (y_1) 
(y_1) edge node {} (y_2)  
(y_2) edge node {} (y_3) 
(z_0) edge node {} (z_1) 
(z_1) edge node {} (z_2) 
(z_2) edge node {} (z_3)
(x_0) edge node {} (y_0) 
(y_0) edge node {} (z_0) 
(x_1) edge node {} (y_1) 
(y_1) edge node {} (z_1)
(x_2) edge node {} (y_2) 
(y_2) edge node {} (z_2)   
(x_{k-1}) edge node {} (y_{k-1}) 
(y_{k-1}) edge node {} (z_{k-1}) 
(x_{k-1}) edge node {} (x_k) 
(z_{k-1}) edge node {} (z_k)
(z_k) edge node {} (z_{k+1})
(z_{k+1}) edge node {} (y_{k+1})
(z_{k+2}) edge node {} (y_{k+2})
(x_{k-2}) edge node {} (x_{k-1})
(y_{k-2}) edge node {} (y_{k-1})
(z_{k-2}) edge node {} (z_{k-1});

\draw[thick, mid arrow] (y_{k-1}) to (y_k);
\draw[thick, mid arrow] (z_k) to (y_k);
\draw[thick, mid arrow] (y_k) to (x_k);
\draw[thick, mid arrow] (x_k) to (x_{k+1});
\draw[thick, mid arrow] (y_{k+1}) to (y_k);
\draw[thick, mid arrow] (y_{k+1}) to (x_{k+1});
\draw[thick, mid arrow] (x_{k+1}) to (x_{k+2});
\draw[thick, mid arrow] (x_{k+2}) to (x_{k+3});
\draw[thick] (y_{k+1}) to (y_{k+2});
\draw[thick] (z_{k+1}) to (z_{k+2});
\draw[thick, mid arrow] (y_{k+2}) to (x_{k+2});
\draw[thick] (y_{k+2}) to (y_{k+3});
\draw[thick] (z_{k+2}) to (z_{k+3});

\end{tikzpicture}

\caption{The graph in Example~\ref{ex:instability}~(3).
The arrows indicate the directions of increase for $u$.}
\label{fig:ladder}
\end{figure}

\end{example}

\appendix

\section{Form uniqueness and capacity of the Cauchy boundary}\label{s:appendix}
In this appendix, we give more details on the connection
between capacity of the Cauchy boundary and form uniqueness
for locally finite graphs. This material
is used in Subsection~\ref{ss:capacity} and all 
definitions are given there.
In particular,
we extend some arguments given in \cite{HKMW13} for the case
of $c=0$ to that of a general $c$ and sketch a proof
of Theorem~\ref{t:capacity_general} which states
that if the capacity of the Cauchy boundary is finite,
then it is zero if and only if the graph satisfies form uniqueness,
see also Theorem~6.9 in \cite{Sch20b}.

We denote the closure of $X$ with respect to a path metric $d_\si$
as $\overline{X}^\si$ and let $\pt_\si X=\overline{X}^\si \setminus X$
denote the Cauchy boundary. We now restate the main result
for the convenience of the reader.

\begin{theorem}[Theorem~\ref{t:capacity_general} in Subsection~\ref{ss:capacity}]\label{at:capacity_general}
Let $(b,c)$ be a locally finite graph over $(X,m)$. Let $d_\si$
be a strongly intrinsic path metric with Cauchy
boundary $\pt_\si X$. If $\cp(\pt_\si X)<\infty$, then
the following statements are equivalent:
\begin{itemize}
\item[(i)] $\QD \neq \QN$.
\item[(ii)] $0< \cp(\pt_\si X)$.
\end{itemize}
\end{theorem}

We start the proof by recalling some general theory.
If $U \subseteq \overline{X}^\si$ is open with $\cp(U)<\infty$, then there 
exists a unique $e \in D(\QN)$ with 
$$0 \le e \le 1, \qquad e = 1 \textup{ on } U \cap X\textup{,} \qquad
\cp (U)= \|e\|^2_{\Q}.$$
See Lemma~2.1.1 in \cite{FOT94} for the existence of such an $e$
for general Dirichlet forms and Lemma~4.1 in \cite{HKMW13}
for details on applying this to the graph case. We call $e$ the \emph{equilibrium
potential associated to $U$}. 

We next prove a general lemma that says that if a graph satisfies form 
uniqueness, then the capacity of the boundary is either zero or infinity.
We do not require local finiteness
or that the path metric satisfies any particular condition at this point.
\begin{lemma}\label{l:capacity_alternative}
Let $(b, c)$ be a graph over $(X, m)$. Let $d_\si$ be a path metric with Cauchy
boundary $\pt_\si X$. 
If $\QD=\QN$, then either $\cp(\pt_\si X)=0$ or $\cp(\pt_\si X)=\infty$.
\end{lemma}
\begin{proof}
Let $\QD=\QN$ and $\cp(\pt_\si X)<\infty$. 
Then, there exists an open set $U$ in $\overline{X}^\si$ such that
$\pt_\si X \subseteq U$ and $\cp(U)<\infty$.
Let $e \in D(\QN)$ be the equilibrium potential associated to $U$.
Since $\QD=\QN$, there exists a sequence $e_n \in C_c(X)$ 
such that $e_n \to e$ in $\|\cdot\|_{\Q}$.
Since $e - e_n \ge 1$ on $U_n \cap X$ for some neighborhood 
$U_n$ of $\partial_\sigma X$ in $\overline{X}^\sigma$, we have
\[
\cp (\pt_\si X) \le \lim_{n \to \infty} \|e-e_n\|_{\Q} =0.
\]
This completes the proof.
\end{proof}

Next, we show that capacity zero implies form uniqueness.
Recall that a path metric $d_\si$ is strongly intrinsic if
$$\sum_{y \in X}b(x,y) \si^2(x,y) \leq m(x)$$
for all $x \in X$. 
Note that we only require the metric
to be strongly intrinsic and the graph to be locally finite at the last step.

\begin{lemma}\label{l:polar_capacity}
Let $(b, c)$ be a locally finite graph over $(X, m)$.
Let $d_\si$ be a path metric with Cauchy
boundary $\pt_\si X$ such that $\cp(\partial_\sigma X) = 0$.  
Let $u \in D(\QN)$ and $\epsilon > 0$.  
Then there exists an open set $U \subseteq \overline{X}^\sigma$ such that $\partial_\sigma X \subseteq U$,
$\cp(U)<\infty$, and
\[
\| u -(1 - e) u \|_\Q < \epsilon
\]
where $e$ is the equilibrium potential associated to $U$.

If $d_\sigma$ is additionally a strongly intrinsic path metric,
then there exists $\varphi \in C_c(X)$ such that
\[
\| (1 - e) u -\varphi \|_\Q < \epsilon.
\]
In particular, $\QD=\QN$.
\end{lemma}

\begin{proof} Since $D(\QN) \cap \ell^\infty(X)$ is dense in $D(\QN)$
with respect to $\|\cdot\|_\Q$, see Lemma~3.16 in \cite{KLW21}, 
we assume throughout that $u$ is bounded.

Let $e_n$ be the equilibrium potential associated to $U_n$ where
$U_n$ are open sets in $\overline{X}^\sigma$ such that $\partial_\sigma X \subseteq U_n$ and $\cp(U_n) = \| e_n\|_{\Q} \to 0$ as $n \to \infty$. 
We then estimate $\|u-(1-e_n)u\|_{\Q} = \|e_n u\|_{\Q}$ as follows:
\begin{align*}
\|e_n u\|_{\Q}^2 
\le & \sum_{x,y \in X} b(x,y)  e_n^2(x)(u(x)-u(y))^2 +
\|u\|^2_\infty \sum_{x,y \in X} b(x,y) (e_n(x)-e_n(y))^2 \ \\
& + \|u\|^2_\infty \sum_{x \in X} \left( c(x)+m(x) \right)e_n^2(x) \\
\leq &\sum_{x,y \in X} b(x,y) e_n^2(x)(u(x)-u(y))^2 + 2 \|u\|^2_\infty \cp (U_n).
\end{align*}
Since $e_n(x)\to0$ for all $x \in X$, the dominated convergence theorem now
implies $\|e_n u\|_\Q= \|u-(1-e_n)u\|_\Q \to0$ as $n \to \infty$.
This shows the first part of the lemma.

Now, assume that $d_\sigma$ is a strongly intrinsic path metric.
Let $Y = X \setminus U$ where $U$ is an open set in $\overline{X}^\si$
such that $\partial_\sigma X \subseteq U$
and consider the restriction
of $(b,c)$ to $Y$, denoted by $(b_Y,c_Y)$.
Assume that $(b_Y,c_Y)$ is connected and let 
$d_Y$ be the distance given by restricting $\sigma$ to edges
defined by $b_Y$.
We claim that $(Y, d_Y)$ is complete.

This follows as $d_Y \ge d_\si$ on $Y$ gives
that any Cauchy sequence $(x_n)$ 
of vertices in $Y$ 
with respect to $d_Y$ is also Cauchy with respect to $d_\si$.  
However, the limit of $(x_n)$ does not belong to $Y$ and eventually $(x_n)$ becomes constant since $d_Y$ is a discrete metric by local
finiteness and the use of path metrics, see Theorem~A.1 in \cite{HKMW13}. 
This completes the proof of the claim.

Now, fix $x_0 \in Y$ and define
\[
\eta_r(x) = \left( \frac{2r - d_Y(x, x_0)}{r} \right)_+ \wedge 1
\]
for $r \in \N_0$ and $x \in Y$.
Since $Y$ is complete, all balls defined with respect to $d_Y$
are finite, see Theorem~A.1 in \cite{HKMW13}. Thus, $\eta_r \in C_c(Y)$
as $\eta_r$ is supported on a ball of radius $2r$. 
Furthermore, $\eta_r \to 1$ as $r \to \infty$.
Finally, by the triangle inequality and the strong
intrinsic condition, we get
$$\sum_{y \in X} b(x,y)(\eta_r(x)-\eta_r(y))^2 \leq \frac{m(x)}{r^2}.$$

Let $u \in D(\QN) \cap \ell^\infty(X)$
and let $U \subseteq \overline{X}^\sigma$ be open with $\partial_\sigma X \subseteq U$
and $\cp(U)<\infty$.
Let $e$ be the equilibrium potential associated to $U$.
Let $v=(1-e)u$ and $\varphi_r = \eta_r v \in C_c(Y)$. 
We now show that $\varphi_r$ converges to $v$
 in $\|\cdot\|_\Q$ which will complete the proof.

Set $h_r = v - \varphi_r = (1-\eta_r)v = (1-\eta_r)(1-e)u$.
Taking into account that $h_r=0$ on $U \cap X$, we have
\[
\QN (h_r) = \frac{1}{2}\sum_{x,y \in Y} b(x,y)(h_r(x)-h_r(y))^2
+\sum_{x \in Y} \sum_{y \in U} b(x,y) h_r^2(x) + \sum_{x \in Y} c(x)h_r^2(x).
\]
For the first term, we use basic estimates and the strong
intrinsic condition to get:
\begin{align*}
& \frac{1}{2}\sum_{x,y \in Y} b(x,y)(h_r(x)-h_r(y))^2 \\
%\le
%&
%\sum_{x,y \in Y} b(x,y)
%\left( (1-\eta_r (x))^2 (v(x)-v(y))^2 + v^2(y) (\eta_r(x)-\eta_r(y))^2\right) %\\
%=
\le &
\sum_{x,y \in Y} b(x,y)(1-\eta_r (x))^2 (v(x)-v(y))^2 + 
\sum_{x \in Y} v^2(x){\sum_{y \in Y} b(x,y){(\eta_r(x) -\eta_r(y))^2}}\\
\leq
&
\sum_{x,y \in Y} b(x,y)(1-\eta_r (x))^2 (v(x)-v(y))^2 + \frac{\|v\|^2}{r^2} \to 0
\end{align*}
as $r \to \infty$ by the dominated convergence theorem
since $\eta_r \to 1$ as $r \to \infty$.
Recall that $e=1$ on $U \cap X$ and use this to bound
the remaining terms by
$$\|u\|^2_\infty \sum_{x \in Y} \sum_{y \in U}b(x,y) (1-\eta_r(x))^2(e(x)-e(y))^2 + \sum_{x \in Y} c(x) (1-\eta_r(x))^2 v^2(x)$$
which also tend to zero as $r \to \infty$, again by the dominated convergence theorem.
Therefore, $\|h_r\|_\Q \to 0$, and thus $\varphi_r \to v$ in $\|\cdot\|_Q$.
This completes the proof when $Y$ is connected.

If $Y$ is not connected, then there exist at most countably many connected components $\{Y_i\}_{i \in \N_0}$ of $Y$.  
For $v = (1 - e)u$, we then have
\[
\QN(v) = \sum_{i \in \N_0} \QN(v \cdot 1_{Y_i})
\]
where $1_{Y_i}$ is the characteristic function of $Y_i$.
Thus, the conclusion follows by applying the above argument to each $Y_i$ and 
by a diagonal sequence argument.
\end{proof}

\begin{proof}[Proof of Theorem~\ref{at:capacity_general}]
Lemmas~\ref{l:capacity_alternative}~and~\ref{l:polar_capacity} together 
prove the theorem.
\end{proof}

\subsection*{Acknowledgments}
J.~M.~acknowledges the generous hospitality of the Graduate
Center and York College of the City University of New York.
R.~K.~W.~acknowledges the generous hospitality of Tohoku University.
G.~R.~acknoweledges the generous hospitality of Cornell University and 
wishes to thank the organizers of the 8th Cornell Conference on Analysis, Probability, 
and Mathematical Physics on Fractals for the opportunity to present this work.

\end{document}